\pdfoutput=1
\RequirePackage{amsmath}

\documentclass[a4paper,runningheads,final]{llncs}

\usepackage[utf8]{inputenc}
\usepackage[USenglish]{babel}
\usepackage[T1]{fontenc}
\usepackage{hyperref}

\usepackage{times}
\usepackage{microtype}

\usepackage{paralist}
\usepackage{enumitem}

\usepackage{mathtools}
\usepackage{amssymb}
\usepackage{bussproofs}
\usepackage{ifthen}
\usepackage{sty/math}

\usepackage{sty/comments}

\usepackage{tabularx}

\usepackage{listings}
\lstdefinestyle{default}{
	language=Java,
	basicstyle=\footnotesize\ttfamily,
	keywordstyle=\color{teal}\ttfamily,
	commentstyle=\color{gray}\ttfamily,
    tabsize=2,
	numberbychapter=false,
	morekeywords={NULL,pointer_t,data_t,CAS,atomic,random_value,global,local,malloc,free,assume,udef,havoc,STOP,struct,shared,uint},
    moredelim=**[is][\bfseries]{@}{@},
    keepspaces=true,
}
\lstdefinestyle{condensed}{
	style=default,
	basicstyle=\scriptsize\ttfamily,
}
\lstdefinestyle{inline}{
	style=default,
	basicstyle=\ttfamily,
	keywords={},
}
\usepackage{xspace}
\newcommand{\eg}{e.\kern+.75ptg.\xspace}
\newcommand{\ie}{i.\kern+.75pte.\xspace}

\usepackage[capitalise,nameinlink]{cleveref}
\crefname{assumption}{Assumption}{Assumptions}
\Crefname{assumption}{Assumption}{Assumptions}
\crefname{fact}{Fact}{Facts}
\Crefname{fact}{Fact}{Facts}
\crefname{figure}{Figure}{Figures}
\Crefname{figure}{Figure}{Figures}

\spnewtheorem{assumption}{Assumption}{\bfseries}{\itshape}
\spnewtheorem{fact}{Fact}{\bfseries}{\itshape}

\title{%
	Effect Summaries for Thread-Modular Analysis
}
\subtitle{%
	\textnormal{Sound Analysis despite an Unsound Heuristic}
}

\author{%
	Luk\'a\v s Hol\'ik\hspace{.5pt} \inst{1}
	\and
	Roland Meyer\hspace{1pt} \inst{2}
	\and
	Tom\'a\v s Vojnar\hspace{.85pt} \inst{1}
	\and
	Sebastian Wolff\hspace{1.75pt} \inst{2}\fnmsep\inst{3}
}

\institute{%
	\textsuperscript{1} Brno University of Technology
	\quad
	\textsuperscript{2} TU Braunschweig
	\quad
	\textsuperscript{3} Fraunhofer ITWM
}

\def\githuburl{%
	https://github.com/Wolff09/TMRexp/releases/SAS17/
}

\begin{document}
	\maketitle
	%!TEX root = ../main.tex

\begin{abstract} We propose a novel guess-and-check principle to increase the
efficiency of thread-modular verification of lock-free data structures. We build
on a heuristic that guesses candidates for stateless effect summaries of
programs by searching the code for instances of a~copy-and-check programming
idiom common in lock-free data structures. These candidate summaries are used to
compute the interference among threads in linear time. Since a candidate summary
need not be a sound effect summary, we show how to fully automatically check
whether the precision of candidate summaries is sufficient. We can thus perform
sound verification despite relying on an unsound heuristic. We have implemented
our approach and found it up to two orders of magnitude faster than existing
ones. \end{abstract}

	%!TEX root = ../main.tex

%==========================================================================
\section{Introduction}
\label{sec:introduction}
%==========================================================================

Verification of concurrent, lock-free data structures has recently received
considerable attention
\cite{DBLP:conf/tacas/AbdullaHHJR13,DBLP:conf/sas/AbdullaJT16,DBLP:conf/vmcai/HazizaHMW16,DBLP:conf/ecoop/PintoDG14,DBLP:journals/tocl/SchellhornDW14}.
Such structures are both of high practical relevance %highly practically relevant
and, at the same time,
difficult to write. A common correctness notion in this context is
\emph{linearizability}~\cite{DBLP:journals/toplas/HerlihyW90}, which requires
that every concurrent execution can be linearized to an execution that could
also occur sequentially. For many data structures, linearizability reduces to
checking control-flow reachability in a variant of the data structure that is
augmented with observer automata~\cite{DBLP:conf/tacas/AbdullaHHJR13}.  This
control-flow reachability problem, in turn, is often solved by means of
\emph{thread-modular analysis}
\cite{DBLP:conf/cav/BerdineLMRS08,DBLP:journals/toplas/Jones83}.  Our
contribution is on improving thread-modular analyses for verifying
linearizability of lock-free data structures.

%%%%%%%%%%%%%%%%%%%%%%%%%%%%%%%%%%%%%%%%%%%%%%%%%%%%%%%%%%%%%%%%%%%%%%%%%%%%%%
%%%%%%%%%%%%%%%%%%%%%%%%%%%%%%% THREAD-MODULAR %%%%%%%%%%%%%%%%%%%%%%%%%%%%%%%
%%%%%%%%%%%%%%%%%%%%%%%%%%%%%%%%%%%%%%%%%%%%%%%%%%%%%%%%%%%%%%%%%%%%%%%%%%%%%%

Thread-modular analyses compute the least solution to a recursive
equation\vspace*{-1mm}
\[
   X ~=~ X \,\cup\, \mathit{seq}(X) \,\cup\, \mathit{interfere}(X) 
   \,.\vspace*{-1mm}
\]
The domain of $X$ are sets of \emph{views}, partial configurations reflecting
the perception of a~single thread about the shared heap. Crucially,
thread-modular analyses abstract away from the correlation among the views of
different threads. Function $\mathit{seq}(X)$ computes a sequential step,
the views obtained from $X$ by letting each thread execute a~command on its own
views. This function, however, does not reflect the fact that a thread may
change a part of the shared heap seen by others. Such interference steps are
computed by $\mathit{interfere}(X)$. It is this function that we improve on.
Before turning to the contribution, we recall the existing approaches and
motivate the need for more work.

%
% merge-and-project
%

In the \emph{merge-and-project} approach to interference (e.g.,
\cite{DBLP:conf/cav/BerdineLMRS08,DBLP:conf/spin/FlanaganQ03,DBLP:journals/toplas/Jones83,DBLP:conf/ictac/MalkisPR06}),
a merge operation is applied on every two views in $X$ to determine all merged
views consistent with the given ones. On each of the consistent views, some
thread performs a sequential step, and the result is projected to what is seen
by the other threads. The approach has problems with efficiency. The number of merge
operations is exactly the square of the number of views in the fixed point. In
addition, every merge of two views is expensive. It has to consider all
consistent views whose number can be exponential in~the~size~of~the~views.

%
% learning
%

The \emph{learning} approach to interference
\cite{DBLP:conf/vmcai/Vafeiadis10,DBLP:conf/esop/Mine11} derives, via symbolic
execution, a~symbolic update pattern for the shared heap. The learning process
is integrated into the fixed-point computation, which incurs an overhead.
Moreover, the number of update patterns to be learned is bounded by the number
of reachable views only. An interference step applies the learned update
patterns to all views, which again is quadratic in the number of views.
Moreover, although update patterns abstract away from thread-local information,
computing each application still requires a potentially expensive matching.
There are, however, fragments of separation logic with efficient
entailment~\cite{DBLP:conf/concur/CookHOPW11}.

%%%%%%%%%%%%%%%%%%%%%%%%%%%%%%%%%%%%%%%%%%%%%%%%%%%%%%%%%%%%%%%%%%%%%%%%%%%%%%
%%%%%%%%%%%%%%%%%%%%%%%%%%%%%%%% CONTRIBUTION %%%%%%%%%%%%%%%%%%%%%%%%%%%%%%%%
%%%%%%%%%%%%%%%%%%%%%%%%%%%%%%%%%%%%%%%%%%%%%%%%%%%%%%%%%%%%%%%%%%%%%%%%%%%%%%

What is missing is an \emph{efficient} approach to computing interferences among
threads.  

\smallskip

\noindent\textbf{Main ideas of the contribution.} We propose to
compute $\mathit{interfere}(X)$ by means of so-called \emph{effect summaries}.
An effect summary for a method $\Method$ is a \emph{stateless} program
$\SummaryMethod$ which over-approximates the effects that $\Method$ has on the
shared heap. With such summaries at hand, the interference step can be computed
in linear time by executing the method summaries $\SummaryMethod$ for
all methods $M$ on the views in the current set $X$. This is a substantial
improvement in efficiency over merge-and-project and learning techniques, which
require time roughly quadratic in the size of the fixed-point approximant, $X$,
and possibly exponential in the size of views.

%
% statelessness
%

Technically, \emph{statelessness} is defined as atomicity and absence of
persistent~local~state. We found both requirements typically satisfied by
methods of lock-free data structures. For our approach, this means stateless
summaries are likely to exist (which is confirmed by our experiments). The
reason why the atomicity requirement holds is that the methods have to preserve
the integrity of the data structure under interleavings. The absence of
persistent state holds since interference by
other threads may invalidate local~state~at~any~time.
%Hence, relying on such state could cause corruption of the data structure.

%
% heuristic
%

We propose a heuristic to compute, from a method $\Method$, a stateless
program~$\SummaryMethod$ which is a candidate for being an effect summary of
$M$. Whether or not this candidate is indeed a summary of $\Method$ is checked
on top of the actual analysis, as discussed below. Our heuristic is based on
looking for occurrences of a programming idiom common in lock-free data
structures which we call \emph{copy-and-check blocks}. Such a block is a piece
of code that, despite lock-free execution, appears to be executed atomically.
Roughly, we identify each such block and turn it into an atomic program.

%
% copy-and-check, CAS block
%

Programmers achieve the above mentioned atomicity of copy-and-check blocks
by first creating a local copy of a shared variable, performing some
computation over it, checking whether the copy is still up-to-date and, if so,
publishing the results of the computation to the shared heap. A classic
implementation of such blocks is based on \emph{compare-and-swap} (CAS)
instructions. In this case, for a local variable \code{t} and a shared variable
\code{T}, the copy-and-check block typically starts with an assignment
\code{t=T} and finishes with executing \code{CAS(T,t,x)} which atomically checks
whether \code{t==T} holds and, if so, changes the value of \code{T} to \code{x}.
Hereafter, we will denote such blocks as \emph{CAS blocks}, and we will
concentrate on them since they are rather common in practice. However, we note
that the same principle can be used to handle other kinds of copy-and-check
blocks, \eg, those based on the \emph{load-link/store-conditional} (LL/SC)
mechanism.

%
% soundness
%

The idea of program analysis exploiting the intended atomicity of CAS blocks by
treating these blocks as atomic is quite natural. The reason why it is not
common practice is that this approach is not sound in general. The atomicity may
be introduced too coarsely, and, as a result, an $\mathit{interfere}(X)$
implementation based on the guessed candidate summaries may miss interleavings
present in the actual program. For our analysis, this means that its soundness
is conditional upon the fact that the candidate summaries used are indeed proper
effect summaries. It must be checked that they are stateless and that they cover
all effects on the shared heap. We propose a fully automatic and efficient way
of performing those checks. To the best of our knowledge, we are the first to
propose such checks.

%
% effect coverage
%

To check whether candidate summaries indeed cover the effects of the methods for
which they were constructed, the idea is to let the methods execute under any
number of interferences with the candidate summaries and see if some effect not
covered by the candidate summaries can be obtained. Formally, we use the program
$\Summary=\bigoplus_i\SummaryMethod[_i]$, which executes a non-deterministically
chosen candidate summary $\SummaryMethod[_i]$ of a method $\Method_i$, execute
the Kleene iteration $\SummaryStar$ in parallel with each method $M$, and check
whether the following inclusion holds:\vspace*{-1mm}
\[ \mathit{Effects}(\Method\parallel \Summary^*) ~\subseteq~ 
   \mathit{Effects}(\Summary^*).\vspace*{-1mm}
\]
If this inclusion
holds, $\SummaryStar$ covers the actual interference all methods may cause.
Hence, our novel implementation of $\mathit{interfere}(X)$ explores all possible
interleavings. The cost of the inclusion test is asymptotically covered by that
of computing the fixed point, and practically negligible. It can be checked in
linear time (in the size of the fixed point) by performing, for every view in
$X$, a sequential step and testing whether the effect of the step can be
mimicked by the candidate summaries. It is worth pointing out the cyclic nature
of our reasoning: we use the candidate summaries to prove their own correctness.

%
% statelessness
%

Statelessness is an important aspect in the above process. It guarantees that
the sequential iteration of $\SummaryStar$ explores the overall interference the
methods of the data structure cause. As we are interested in parametric
verification, the overall interference is, in fact, the one produced by an
unbounded number of concurrent method invocations. Hence, computing this
interference using candidate summaries requires us to analyse the program
$\prod^\infty \Summary$, which is a parallel composition of arbitrarily many
$\Summary$ instances. However, statelessness guarantees that each of these
instances executes atomically without retaining any local state. While, the
atomicity ensures that the concurrent $\Summary$ instances cannot overlap, the
absence of local state ensures that $\Summary$ instances cannot influence each
other, even if executed consecutively by the same thread. Hence, we can use a
single thread executing the iteration $\SummaryStar$ in order to explore the
interference caused by $\prod^\infty \Summary$. This justifies the usage of
$\SummaryStar$ for the effect coverage above. The check for statelessness is
similar to the one of effect coverage. If both tests succeed, the analysis
information is guaranteed to be sound.

%
% Overview (recap)
%

\smallskip

\noindent\textbf{Overview of the approach, its advantages, and experimental
evaluation.} 
Overall, our thread-modular analysis proceeds as follows. We employ the CAS
block heuristic to compute candidate summaries. We use these candidates to
determine the interferences in the fixed-point computation. Once the fixed point
has been obtained, we check whether the candidates are valid summaries. If so,
the fixed point contains sound information, and can be used for verification
(or, an on-the-fly computed verification result can be used). Otherwise,
verification fails. Currently, we do not have a~refinement loop because it was
not needed in our experiments.

%
% Advantages
%

Our method overcomes the limitations in the previous approaches as follows.
The summary program, $\Summary$, is quadratic in the syntactic size of the
program---not in the size of the fixed point. The interference step executes the
summary on all views in the current set $X$, which means an effort linear rather
than quadratic in the fixed-point approximant. Moreover, $\Summary$ is often
acyclic and hence needs linear time to execute, as opposed to the worst case
exponential merge or match. In our benchmarks, we needed at most 5 very short
summaries, usually around 3--5 lines of code each. The computation of candidate
summaries (based on cheap and standard static analyses) and their check for
validity are separated from the fixed point, and the cost of both operations is
negligible.

%
% Experiments
%

We implemented our thread-modular analysis with effect summaries on top of our
state-of-the-art tool \cite{DBLP:conf/vmcai/HazizaHMW16} based on thread-modular
reasoning with merge-and-project. We applied the implementation to verify
linearizability in a number of concurrent list implementations. Compared to
\cite{DBLP:conf/vmcai/HazizaHMW16}, we obtain a speed-up of two orders of
magnitude. Moreover, we managed to infer stateless effect summaries for all our
case studies except the DGLM queue~\cite{DBLP:conf/forte/DohertyGLM04} under
explicit memory management (where one needs to go beyond statelesness). However,
we are not aware of any automatic approach that would be able to handle this
algorithm.

%==========================================================================

	%!TEX root = ../main.tex

%===============================================================================
\vspace{-3mm}
\section{Effect Summaries on an Example}
\label{Section:Illustration}
\vspace{-2mm}
%===============================================================================

The main complication for writing lock-free algorithms is to guarantee
robustness under interleavings. The key idea to tackle this issue is to use a
specific update pattern, namely the CAS-blocks discussed in
\cref{sec:introduction}. We now show how CAS blocks are employed in the
Treiber's lock-free stack implementation under garbage collection, the code of
which is given in \cref{code:treibers}. The \code{push} method implements a CAS
block as follows:
\begin{inparaenum}[(1)]
  \item copying the top of stack pointer, \code{top=ToS},
  \item linking the \code{node} to be inserted to the current top of stack, 
    \code{node.next=top}, and
  \item making \code{node} the new top of stack in case no other thread changed 
    the shared state, \code{CAS(ToS,top,node)}.
\end{inparaenum}
Similarly, \code{pop} proceeds by:
\begin{inparaenum}[(1)]
  \item copying the top of stack pointer, \code{top=ToS},
  \item querying its successor, \code{next=top.next}, and
  \item swinging \code{ToS} to that successor in case the stack did not change, 
    \code{CAS(ToS,top,next)}.
\end{inparaenum}

%!TEX root = ../../main.tex

\begin{figure}[t]
\vspace{-3mm}
	\center
	% Code
	\begin{minipage}{.92\textwidth}
		\begin{minipage}{.52\textwidth}
			\begin{lstlisting}[style=condensed]
struct Node { data_t data; Node next; }
shared Node ToS;

void push(data_t in) {
	Node node = new Node(in);
	while (true) {
		Node top = ToS;
		node.next = top;
		if(CAS(ToS, top, node)){
			return;
}	}	}

S1: atomic {
	/* push */
	Node node = new Node(*);
	node.next = ToS;
	ToS = node;
}
			\end{lstlisting}
		\end{minipage}
		\hfill
		\begin{minipage}{.4\textwidth}
			\begin{lstlisting}[style=condensed]
bool pop(data_t& out) {
	while (true) {
		Node top = ToS;
		if(top == NULL){
			return false;
		}
		Node next = top.next;
		if(CAS(ToS, top, next)){
			out = top.data;
			return true;
}	}	}

S2: atomic {
	/* pop */
	assume(ToS != NULL);
	ToS = ToS.next;
}
S3: atomic { /* skip */ }
			\end{lstlisting}
		\end{minipage}
	\end{minipage}
	\vspace{-5mm}
	\begin{lstlisting}[caption={Pseudo code of the Treiber's lock-free stack \cite{opac-b1015261} and its effect summaries.},label={code:treibers}]
	\end{lstlisting}
	\vspace{-12mm}
\end{figure}

Following the CAS-block idiom, the only statements modifying the shared heap in
the Treiber's stack are the CAS operations. Hence, we identify three types of
effects on the shared heap. First, a successful CAS in \code{push} makes
\code{ToS} point to a newly allocated cell that, in turn, points to the previous
value of \code{ToS}. Second, a successful CAS in \code{pop} moves \code{ToS} to
its successor \code{ToS.next}. Since we assume garbage collection, the removed
element is not freed but remains in the shared heap until collected. Third, the
effect of any other statement on the shared heap is~the~identity.

With the effects of the Treiber's stack identified, we can turn towards finding
an approximation. For that, consider the program fragments from
\cref{code:treibers}: \code{S1} covers the effects of the CAS in \code{push},
\code{S2} covers the effects of the CAS in \code{pop}, and, lastly, \code{S3}
produces the identity-effect covering all remaining statements. Then, the
summary program is $\Summary=\mcode{S1}\oplus\mcode{S2}\oplus\mcode{S3}$.

To obtain the non-trivial summaries \code{S1} and \code{S2}, it suffices to
concentrate on the block of code between the \code{top=ToS}
assignment and the subsequent \code{CAS(ToS,top,_)}
statement. Without going into details (which will be provided in
\cref{sec:computing-effect-summaries}), the summaries result from considering
the code between the two statements atomic, performing simplification of the
code under this atomicity assumption, and including some purely local
initialization and finalization code (such as the allocation in the
\code{push} method).

	%!TEX root = ../main.tex

%==========================================================================
\vspace{-2mm}
\section{Programming Model} \label{sec:programming-model}
\vspace{-2mm}
%==========================================================================

%%%%%%%%%%%%%%%%%%%%%%%%%%%%%%%%%%%%%%%%%%%%%%%%%%%%%%%%%%%%%%%%%%%%%%%%%%%%%%
%%%%%%%%%%%%%%%%%%%%%%%%%%%%%%%%%%%% HEAP %%%%%%%%%%%%%%%%%%%%%%%%%%%%%%%%%%%%
%%%%%%%%%%%%%%%%%%%%%%%%%%%%%%%%%%%%%%%%%%%%%%%%%%%%%%%%%%%%%%%%%%%%%%%%%%%%%%

A concurrent program $\Program$ is a parallel composition of threads $\Thread$.
The threads are while-programs formed using sequential composition,
non-deterministic choice, loops, atomic blocks, skip, and primitive commands.
The syntax is as follows:
\begin{gather*}	
\Program ::= ~
	\Thread \EBNFalt
	\Program\inpar\Program
\qquad
\Thread ::= ~
	\Thread_1;\Thread_2 \EBNFalt
	\Thread_1\oplus \Thread_2 \EBNFalt
	\Thread^* \EBNFalt
	\atomici{\Thread} \EBNFalt
	\cmdskip \EBNFalt
	% \cmderr \EBNFalt
	\EBNFprimitive
.
\end{gather*}
We use $\Thrd$ for the set of all threads. We also write $\Program^*$ to mean a
program $\Program$ with the Kleene star applied to all threads. The syntax and
semantics of the commands in $\EBNFprimitive$ are orthogonal to our development.
We comment on the assumptions we need in a moment.

We assume programs whose threads implement methods from the interface of the
lock-free data structure whose implementation is to be verified. The fact that,
at runtime, we may find an arbitrary (finite) number of instances of each of the
threads corresponds to an arbitrary number of the methods invoked concurrently.
The verification task is then formulated as proving a designated shared heap
unreachable in all instantiations of the program. Since thread-modular analyses
simultaneously reason over all instantiations of the program, we refrain from
making this parameterization more explicit. Instead, we consider program
instances simply as programs with more copies of the same threads.

%%%%%%%%%%%%%%%%%%%%%%%%%%%%%%%%%%%%%%%%%%%%%%%%%%%%%%%%%%%%%%%%%%%%%%%%%%%%%%
%%%%%%%%%%%%%%%%%%%%%%%%%%%%% STATE + SEPARATION %%%%%%%%%%%%%%%%%%%%%%%%%%%%%
%%%%%%%%%%%%%%%%%%%%%%%%%%%%%%%%%%%%%%%%%%%%%%%%%%%%%%%%%%%%%%%%%%%%%%%%%%%%%%

We model heaps as partial and finite functions
$h\mkern-5mu:\mkern-4mu\Var\cup\Nat\nrightarrow\Nat$. Hence, we do not
distinguish between the stack and the heap, and let the heap provide valuations
for both the program variables from $\Var$ and the memory cells from $\Nat$. We
use $\Heap$ for the set of all heaps. Initially, the heap is empty, denoted by
$\emp$ with $\dom{\emp}=\emptyset$. We write~$\bot$ if a partial function is
undefined for an argument: $h(e)=\bot$ if $e\not\in\dom{h}$.

We assume each thread has an identifier from $\Tid\subseteq \Nat$. A
\emph{program state} is a pair $\state{s}{\cf}$ where $s\in\Heap$ is the shared
heap and $\cf\mkern-5mu:\mkern-4mu\Tid\to\Thrd\times\Heap$ maps the thread
identifiers to \emph{thread configurations}. A thread configuration is of the
form $\conf{\Thread}{o}$ with $\Thread\in\Thrd$ and $o\in\Heap$ being a heap
owned by $T$. If $\cf=\set{i\to\conf{\Thread}{o}}$ contains a single mapping, we
write simply $\state{s}{\conf{\Thread}{o}}$.

Our development crucially relies on having a notion of separation between the shared heap $s$ and the owned heap $o$ of a thread $T$.
However, the actual definitions of what is owned and what shared are a parameter to our development.
We just require the separation to respect disjointness of the shared and owned heaps and to be defined such that it is preserved across execution of program statements.
The latter is formalized below in~\cref{assumption:sequential-separation}.
To render disjointness formally, we say that a~state $\state{s}{\cf}$ is \emph{separated}, denoted by $\wdef{s}{\cf}$, if, for every $i_1,i_2\in\dom{\cf}$ with $\cf(i_j)=\conf{\Thread_{i_j}}{o_{i_j}}$ and $i_1\neq i_2$, we have $\dom{s}\cap\dom{o_{i_j}}\cap\Nat=\emptyset$ and $\dom{o_{i_1}}\cap\dom{o_{i_2}}\cap\Nat=\emptyset$.
Note that, in order to allow for having thread-local variables, the heaps need to be disjoint only on memory cells (but not on variables), thus the additional intersection with $\Nat$.

%%%%%%%%%%%%%%%%%%%%%%%%%%%%%%%%%%%%%%%%%%%%%%%%%%%%%%%%%%%%%%%%%%%%%%%%%%%%%%
%%%%%%%%%%%%%%%%%%%%%%%%%%%%%%%%% TRANSITION %%%%%%%%%%%%%%%%%%%%%%%%%%%%%%%%%
%%%%%%%%%%%%%%%%%%%%%%%%%%%%%%%%%%%%%%%%%%%%%%%%%%%%%%%%%%%%%%%%%%%%%%%%%%%%%%

We use $\stepprog$ to denote \emph{program steps}. The sequential semantics of
threads is as expected for sequential composition, choice, loops, and skip. An
atomic block $\atomici{\,\Thread}$ summarizes a computation of the underlying
thread $\Thread$ into a single program step. The semantics of primitive commands
depends on the actual set $\EBNFprimitive$. We do not make it precise but
require it to preserve separation in the following sense.
\begin{assumption} \label{assumption:sequential-separation} For every step
$\state{s}{\conf{\Thread}{o}} \stepprog\state{s'}{\conf{\Thread'}{o'}}$ with
$\separate{s}{\conf{\Thread}{o}}$, we have $\separate{s'}{\conf{\Thread'}{o'}}$.
\end{assumption}

The semantics of a concurrent program incorporates the requirement for
separation into its transition rule. A thread may only update the shared heap
and those parts of the heap it owns. No other parts can be modified. Therefore,
we let threads execute in isolation and ensure that the combined resulting state
is separated:
\begin{prooftree}
 	\rname{par}
	\raxiom{
		\state{s}{\cf(i)}
		\stepprog
		\state{s'}{\cf''}
	}
	\raxiom{
		\cf'=\cf[i\to\cf'']
	}
	\raxiom{
		\separate{s'}{\cf'}
	}
	\rconclusion{3}{
		\state{s}{\cf}
		\stepprog
		\state{s'}{\cf'}
	}
\end{prooftree}

Despite a precise notion of separation is not needed for the development of our approach in Section~\ref{sec:effect-summaries}, we give, for illustration, the notion we use in our implementation and experiments.
In the case of garbage collection (like in Java), the owned heap of a thread includes, as usual, its local variables and cells accessible from these variables, which were allocated by the thread, but never made accessible through the shared variables.
The shared heap then contains the shared variables, all cells that were once made accessible from them, as well as cells waiting for garbage collection.
For the case with explicit memory management, we need a more complicated mechanism of ownership transfer where a shared cell can become owned again.
We propose such a mechanism in Section~\ref{sec:generalization-to-explicit-memory-management}.

%%%%%%%%%%%%%%%%%%%%%%%%%%%%%%%%%%%%%%%%%%%%%%%%%%%%%%%%%%%%%%%%%%%%%%%%%%%%%%
%%%%%%%%%%%%%%%%%%%%%%%%%%%%%%%%%%%% INIT %%%%%%%%%%%%%%%%%%%%%%%%%%%%%%%%%%%%
%%%%%%%%%%%%%%%%%%%%%%%%%%%%%%%%%%%%%%%%%%%%%%%%%%%%%%%%%%%%%%%%%%%%%%%%%%%%%%

We assume the computation of the programs under scrutiny to start from an initial state $\init{\Program}=\conf{\sinit}{\cfinit{\Program}}$ where $\sinit$ is the result of an initialization procedure.
The initial thread configurations, denoted by $\cfinit{\Thread}$, are of the form $\conf{\Thread}{\emp}$.
The initialization procedure is assumed to be part of the input program.
We are interested in the shared heaps reachable by program $\Program$ from its initial state:
\begin{align*}
  \Reach{\Program} := \setcond{
    s
  }{
    \exists\: \cf.~
      \init{\Program}
      \stepprogany
      \state{s}{\cf}
  }.
\end{align*}

%==========================================================================

	%!TEX root = ../main.tex

%==========================================================================
\section{Interference via Summaries}
\label{sec:effect-summaries}
%==========================================================================

We now present our new approach to computing the effect of thread interference
steps on the shared heap (corresponding to evaluating the expression
$\mathit{interfere}(X)$ from \cref{sec:introduction} for a set of views $X$) in
a way which is suitable for concurrency libraries. In particular, we introduce a
notion of a \emph{stateless effect summary} $\Summary$: a program whose repeated
execution is able to produce all the effects on the shared heap that the program
under scrutiny, $\Program$, can produce. With a stateless effect summary
$\Summary$ at hand, one can compute $\mathit{interfere}(X)$ by repeatedly
applying $\Summary$ on the views in $X$ until a fixed point is reached.  Here,
statelessness assures that $\Summary$ is applicable repeatedly without any need
to track its local state.

Later, in Section~\ref{sec:computing-effect-summaries}, we provide a heuristic
for deriving \emph{candidates} for stateless effect summaries. Though our
experiments show that the heuristic we propose is very effective in practice,
the candidate summary that it produces is not guaranteed to be an effect
summary, i.e., it is not guaranteed to produce all the effects on the shared
heap that $\Program$ can produce. A candidate summary which is not an effect
summary is called \emph{unsound}. To guarantee soundness of our approach even
when the obtained candidate summary is unsound, we provide a test of soundness
of candidate summaries. Interestingly, as we prove, it is the case that even
(potentially) unsound candidate summaries can be used to check their own
soundness---although this step appears to be cyclic reasoning.

%--------------------------------------------------------------------------
\vspace{-2mm}
\subsection{Stateless Effect Summaries}
\vspace{-1mm}
%--------------------------------------------------------------------------

We start by formalizing the notion of statelessness. Intuitively, a thread is
stateless if it terminates after a single step and disposes its local heap.
Formally, we say that a thread $\Thread$ of a program $\Summary$ is
\emph{stateless} if, for all reachable shared heaps $s\in\Reach{\Summary^*}$ and
all transitions $\state{s}{\cfinit{\Thread}}\stepprog\state{s'}{\cf}$, we have
$\cf=\conf{\cmdskip}{\emp}$. A program $\Summary$ is stateless if so are all its
threads. Note that statelessness should hold from all reachable shared heaps
rather than from just all heaps. While an atomic execution to $\cmdskip$ would
be easy to achieve from all heaps, a clean-up yielding $\emp$ can only be
achieved if we have control over the thread-local heap. Also note that
statelessness basically requires a thread to consist of a top-level atomic block
to ensure termination in a single step.

Next, we define the \emph{effects} of a program $\Program$, denoted by
$\effects{\Program}\subseteq\Heap\times\Heap$, to be the set
$\effects{\Program}=\setcond{\pair{s}{s'}}{\init{\Program}\stepprogany\state{s}{\cf}
\stepprog \state{s'}{\cf'}}.$
This set generalizes the reachable shared heaps, $\Reach{\Program}$: it contains
all atomic (single-step) updates $\Program$ performs on the heaps from
$\Reach{\Program}$.

Altogether, a program $\Summary$ is a~\emph{(stateless) effect summary} of
$\Program$ if it is stateless and $\effects{\Thread\inpar\Summary^*} \subseteq
\effects{Q^*}$ holds for all threads $\Thread\in\Program$. We refer to this
inclusion as the \emph{effect inclusion}. Intuitively, it states that
$\Summary^*$ subsumes all the effects $\Thread$ may have under interference with
$\Summary^*$. The below lemma, proved in the appendix, shows that the effect
inclusion can be used to check whether a candidate summary is indeed an effect
summary. Moreover, the check can deal with the different threads separately.

\begin{lemma} \label{new:thm:effect-inclusion-lifting} If $\Summary$ is
stateless and $\effects{\Thread\inpar\Summary^*}\subseteq\effects{\Summary^*}$
holds for all $\mkern+1mu\Thread\mkern-1mu\in\Program\!$, then we have
$\effects{\Program}\subseteq\effects{\Program\inpar\Summary^*}\subseteq\effects{\Summary^*}$.
\end{lemma}

In what follows, we describe our novel thread-modular analysis based on
effect summaries. We assume that, in addition to the program $\Program$ under
scrutiny, we have a program $\Summary$ which is a candidate for being a summary of
$\Program$ (obtained, e.g., by the heuristic that we provide in
Section~\ref{sec:computing-effect-summaries}). In Section~\ref{sec:fixpoint}, we
first provide a fixed point computation where the interference step is
implemented by a repeated application of the candidate summary $\Summary$. We
show that if the candidate summary $\Summary$ is an effect summary, then the
fixed point we compute is a~conservative over-approximation of the reachable
shared heaps of $\Program$. Next, in Section~\ref{sec:soundness-test}, we show
that the fact whether or not $\Summary$ is indeed an effect summary of
$\Program$ can be checked efficiently on top of the computed fixed point (even
though the fixed point need not over-approximate the reachable shared heaps of~$\Program$).

In the case that the test of Section~\ref{sec:soundness-test} fails, $\Summary$
is not an effect summary of $\Program$, and our verification fails with no
definite answer. As a future work, one could think of proposing some way of
patching the summaries based on feedback from the failed test. However, in
our experiments, using the heuristic computation of candidate summaries proposed
in Section~\ref{sec:computing-effect-summaries}, this situation has not happened
for any program where a stateless effect summary exists. In the only experiment
where our approach failed (the DGLM queue under explicit memory management,
which has not been verified by any other fully automatic tool), the notion of
stateless effect summaries itself is not strong enough. Hence, a perhaps
more interesting question for future work is how to further generalize the notion of effect summaries.

%--------------------------------------------------------------------------
\subsection{Summaries in the Fixed-Point Computation}\label{sec:fixpoint}
%--------------------------------------------------------------------------

To explore the reachable shared heaps of program $\Program$, we suggest a thread-modular analysis which explores the reachable states of threads $\Thread\in\Program$ in isolation.
To account for the possible thread interleavings of the original program, we apply interference steps to the threads $\Thread$ by executing the provided summary $\Summary$.
Conceptually, this process corresponds to exploring the state space of the two-thread programs $\Thread\inpar\SummaryStar$ for all syntactically different $\Thread\in\Program$.
Technically, we collect the reachable states of those programs in the following least fixed point:
\begin{align*}
	X_0 &\:=\: \setcond{\state{\sinit}{\conf{\Thread}{\emp}}}{\Thread\in\Program}
	\\
	X_{i+1} &\:=\: X_i \:\cup\: \Post{X_i} \:\cup\: \Env{X_i}\ .
\intertext{%
	Since $\Summary^*$ has no internal state, the analysis only keeps the thread-local configurations of $\Thread$.
	Functions $\Post{\cdot}$ and $\Env{\cdot}$ compute sequential steps (steps~of~$\Thread$) and interference steps (steps of $\Summary^*$), respectively, as follows:
}
	\Post{X_i} &\:=\: \setcond{
		\state{s'}{\cf'}
	}{
		\exists\: \state{s}{\cf}\in X_i.
		~
		\state{s}{\cf}\stepprog\state{s'}{\cf'}
	}
	\\
	\Env{X_i} &\:=\: \setcond{
		\state{s'}{\cf\phantom{'}}
	}{
		\separate{s'}{\cf}
		\:\wedge\:
		\exists\: s,\cf'.
		~
		\\&\qquad\qquad\qquad\;\;
		\state{s}{\cf}\in X_i
		\:\wedge\:
		\state{s}{\cfinit{\Summary}}\stepprog\state{s'}{\cf'}
	}\ .
\end{align*}
Function $\Post{X_i}$ is standard.
For $\Env{X_i}$ we apply $\Summary$ to each configuration $\state{s}{\cf}\in X_i$ by letting it start from the shared heap $s$ and its initial thread-local configuration $\cfinit{\Summary}$.
Then we extract the updated shared heap, $s'$, resulting in the post configuration \state{s'}{\cf}.
Altogether, this procedure applies to the views in $X_i$ the shared heap updates dictated by $\Summary$.
The thread-local configurations, $\cf$, of threads $\Thread$ are not changed by interference.
This locality follows from statelessness. %: $\Summary$ has access only to the shared heap.

The following lemma, which is proven in the appendix, states that the set of shared heaps collected from the above fixed point is indeed the set of reachable shared heaps of all $\Thread\inpar\Summary^*$.
Let $X_k$ be the fixed point and define $\ReachFP=\setcond{s}{\exists\cf.\,\state{s}{\cf}\in X_k}$.
\begin{lemma}
	\label{new:thm:fixed-point-reachset}
	If $\Summary$ is a summary of $\Program$, then
	$\ReachFP = \bigcup_{\Thread\in\Program} \Reach{\Thread\inpar\Summary^*}$.
\end{lemma}

With the state space exploration in place, we can turn towards a soundness result of our method:
given an appropriate summary $\Summary$, the fixed point computation over-approximates the reachable shared heaps of $\Program$.
\begin{theorem}
	\label{new:thm:summary-implies-soundness}
	If $\Summary$ is a summary of $\Program$, then we have $\Reach{\Program}\subseteq\Reach{\Summary^*}=\ReachFP$.
\end{theorem}
The rational behind the theorem is as follows.
Relying on $\Summary$ being a summary of $\Program$ provides the effect inclusion.
So, \cref{new:thm:effect-inclusion-lifting} yields $\effects{\Program\inpar\Summary^*}\subseteq\effects{\Summary^*}$.
From the definition of effects we can then conclude $\Reach{\Program\inpar\Summary^*}\subseteq\Reach{\SummaryStar}$.
Thus, we have $\Reach{\Program}\subseteq\Reach{\Summary^*}$ because $\Reach{\Program}\subseteq\Reach{\Program\inpar\Summary^*}$ is always true.
This shows the first inclusion.
Similarly, the effect inclusion gives $\Reach{\Thread\inpar\Summary^*}\subseteq\Reach{\Summary^*}$ by the definition of reachability.
Hence, we conclude using \cref{new:thm:fixed-point-reachset}.

%--------------------------------------------------------------------------
\subsection{Soundness of Summarization}\label{sec:soundness-test}
%--------------------------------------------------------------------------

Soundness of our method, as stated by \cref{new:thm:summary-implies-soundness} above, is conditioned by $\Summary$ being a summary of $\Program$.
In our framework, $\Summary$ is heuristically constructed and there is no guarantee that it really summarizes $\Program$.
Hence, for our analysis to be sound, we have to check summarization;
we have to establish
\begin{inparaenum}[(1)]
	\item the effect inclusion, and
	\item statelessness~of~$\Summary$.
\end{inparaenum}

To that end, we check that
\begin{inparaenum}[(1)]
	\item every update $\Thread$ performs on the shared heap in the system $\Thread\inpar\Summary^*$ can be mimicked by $\Summary$, and that
	\item every execution of $\Summary$ terminates in a single step and does not retain persistent local state.
\end{inparaenum}
We implement those checks on top of the fixed point, $X_k$, as follows:
\begin{align*}
	\forall\,\state{s}{&\cf}\in X_k ~
	\forall\,s',\cf', i\ \exists \cf''. ~~
	\\
	&\state{s}{\cf}
	\stepprog
	\state{s'}{\cf'}
	\implies
	\conf{s}{\cfinit{\Summary}}
	\stepprog
	\conf{s'}{\cf''}
	\tag{\textsc{chk-mimic}}
	\label{new:fp:check:mimic-effects}
	\;\wedge
	\\
	&\state{s}{\cfinit{\Summary}(i)}
	\stepprog
	\state{s'}{\cf'}
	\implies
	\cf'=\conf{\mathtt{skip}}{\emp}
	% \cf'(i)=\conf{\mathtt{skip}}{\emp}
	~~
	\tag{\textsc{chk-stateless}}
	\label{new:fp:check:stateless}
\end{align*}
The above properties indeed capture our intuition.
The former,  \eqref{new:fp:check:mimic-effects}, states that, for every explored $\Thread$-step %\rfm{Should be ok to have it here.}
of the form $\state{s}{\cf}\stepprog\state{s'}{\cf'}$, the effect $\pair{s}{s'}$ is also an effect of $\Summary$.
That is, executing $\Summary$ starting from $s$ yields $s'$.
This establishes the effect inclusion as required by \cref{new:thm:effect-inclusion-lifting}.
The latter check, \eqref{new:fp:check:stateless}, states that every thread of $\Summary$ must terminate in a single step and dispose its owned heap.
This constraint is relaxed to those shared heaps which have been explored during the fixed-point computation.
That is, it ensures statelessness of $\Summary$ on all heaps from $\ReachFP$.
The key aspect is to guarantee that $\ReachFP$ includes $\Reach{\SummaryStar}$ as required by the definition.
We show that this inclusion follows from the check.

The above checks rely on the fixed point, which, in turn, is computed using the candidate summary $\Summary$.
That is, we use $\Summary$ to prove its own correctness.
Nevertheless, our development results in a sound analysis as stated by the following theorem, the proof of which can be found in the appendix.
\begin{theorem}
	\label{new:thm:checks-iff-summarization}
	The fixed point $X_k$ satisfies \eqref{new:fp:check:mimic-effects} and \eqref{new:fp:check:stateless} if and only if $\Summary$ is a summary of $\Program$.
\end{theorem}

	%!TEX root = ../main.tex

%==========================================================================
\section{Computing Effect Summaries}
\label{sec:computing-effect-summaries}
\vspace{-2mm}
%==========================================================================

We now provide our heuristic for computing effect summaries. 
It is based on CAS blocks between an
assignment \code{t=T}, denoted as \emph{checked assignment}, and a CAS
statement \code{CAS(T,t,x)}\!, denoted as \emph{checking CAS} below.
Since we compute a summary for each such block, the number of summaries
is at most quadratic in the size of the input.

In what follows, consider some method $\Method$ given by its control-flow graph~(CFG)
$G=(V, E, \vinit, \vfinal)$. The CFG has a~unique initial and a~unique final
state, which we will use in our construction. Return commands are assumed to
lead to the final state. As we are only interested in the effect on the shared
heap, we drop return values from return commands. Likewise, we skip assignments
to output parameters unless they are important for the flow of control in $\Method$.
We assume the summaries to execute with non-deterministic input values, and so we replace
every input parameter with a symbolic value $*$. Conditionals, loops, and CAS
commands are represented by two edges, for the successful and failing execution,
respectively. Let $\easgn:=(\vasgns, \mcode{t=T}, \vasgnt)$ be the CFG edge of the
checked assignment, and let the successful branch of the checking CAS be $\ecas
:= (\vcas, \mcode{CAS(T,t,x)}, \vcassuc)$. Next, let $\easgnp := (\vasgns, \mcode{t=T},
\vasgntp)$ be a~copy of the checked assignment to be used as the beginning of
the CAS block, and let $\ecasp := (\vcas, \mcode{CAS(T,t,x)}, \vcassucp)$ be a copy of
the checking CAS to be used as the end of the CAS block. Here, $\vasgntp$ and
$\vcassucp$ are fresh nodes.

To give a concise description of effect summaries, the following shortcuts will be helpful. 
We write $\havocof{G}$ for the CFG obtained from $G$ by replacing
each occurrence of a shared variable by a non-deterministic value $*$. By $G - S$, we mean
the CFG obtained from $G$ by dropping all edges carrying commands from the set
$S$. Given nodes $v_1$ and $v_2$, we denote by $G(v_1,v_2)$ the CFG obtained
from $G$ by making $v_1$/$v_2$ the initial/final node, respectively. Given two
CFGs $G$ and $G'$, we define $G;G'$ as their disjoint union where the single
final state of $G$ is merged with the single initial state of $G'$. Finally, we
allow compositions $e; G$ and $G; e$ of a CFG $G$ with a single edge $e$, by
viewing $e$ as a CFG consisting of a single edge with the initial/final nodes
being the initial/final nodes of $e$,~respectively.

% Given a node $v$ and an edge $e$, we define $G(v, e)$ to be a copy of the CFG
% that starts in the initial node $v$ and replaces the former final node
% $\vfinal$ by a new node $\vfinal'$. The new node is entered via a copy of the
% given edge $e = (v_1, \mathit{command}, v_2)$, namely, $(v_1,
% \mathit{command}, \vfinal')$. Finally, 

% So $x = X$ is mapped to $x = *$ and $\mathit{assume}(\mathit{expr})$ with
% $\mathit{expr}$ containing a shared variable $X$ as well as $cas(X, y, z)$ are
% mapped to $\mathit{assume}(*)$. Moreover, return commands are mapped to
% $\mathit{skip}$.

The construction of the summary proceeds in two steps. First we identify the CAS
block and create the control-flow structure, then we clean it up using data flow
analysis and generate the final code of the summary.
Note that the clean-up step is optional but generates a concise form
beneficial for verification.

% ------------------------------------------------------------------
% ------------------------------------------------------------------
% ------------------------------------------------------------------
% ------------------------------------------------------------------

\paragraph{Step 1: Control-flow structure.}

A summary consists of an initialization phase, followed by the CAS block, and
a finalization phase.
The first step results in the CFG
% The result of the first phase is the CFG
$$\Ginit;\Gblock;\Gfinal\ .
$$
The guiding theme of the construction is to preserve all sequences of commands
that may lead through the CAS block. 

In the initialization phase, which is intended for purely local initialization,
the method is assumed to be interrupted by other threads in the sense that the
values of shared pointers may spontaneously change. Therefore, we replace all
dependencies on shared variables by non-deterministic assignments. Moreover, all return
commands are removed since we have not yet passed the CAS block. Eventually,
when arriving at the $\vasgns$ location, the summary non-deterministically
guesses that the CAS block should begin, and so the control is transferred to it
via the $\easgnp$ edge. Hence, the~initialization~is:
\begin{align*} \Ginit\ :=\ (\havocof{G}-\set{\mcode{return}})(\vinit, \vasgns)\
. 
\end{align*}

The CAS block begins with the $\easgnp$ edge, \ie, with the checked assignment,
and ends with the $\ecasp$ edge, \ie, the checking CAS statement. 
%
% The block is assumed to be executed atomically. 
%
From the CAS block, we remove all control-flow edges with assignments \code{t=T} as
we fixed the checked assignment when entering the CAS block (other assignments
of the form \code{t=T}, if present, will give rise to other CAS blocks; and a
repeated execution of the same checked assignment then corresponds to a repeated
execution of the summary). We also remove the return commands as the
finalization potentially still has to free owned heap. Failing executions of the
checking CAS do not leave the CAS block (and typically get stuck due to the
removed checked assignments). Successful executions may leave the CAS block, but
do not have to. Eventually, the summary guesses the last successful execution of
the checking CAS and enters the finalization phase. Hence, we get the following
code: \begin{align*} \Gblock\ :=\ \easgnp ; ((G - \set{\mcode{return},
\mcode{t=T}})(\vasgnt, \vcas)) ; \ecasp \ .  \end{align*}

Sometimes, the checked assignment can use local
variables assigned prior to the checked assignment. In such a case, we add
copies of edges with these assignments before the $\easgnp$ edge.  This happens,
e.g., in the enqueue procedure of Michael\&Scott's lock-free queue where the sequence
\code{tail=Tail;next=tail.next;} is used. If \code{next=tail.next} is
the checked assignment, we start $\Gblock$ with edges containing \code{tail=Tail}
and \code{next=tail.next}.

The finalization phase, again, cannot rely on shared variables. However, here,
we preserve the return statements to terminate the execution: \begin{align*}
\Gfinal\ :=\ \havocof{G}(\vcassuc, \vfinal)\ . \end{align*}

% -----------------------------------------------------------------------------
% -----------------------------------------------------------------------------
% -----------------------------------------------------------------------------
% -----------------------------------------------------------------------------

\cref{Figure:Phase1PopTreiber} illustrates the construction on the \code{pop}
method in Treiber's stack. Instead of a CFG, we give the source code. \code{STOP}
represents deleted edges and the fact that we cannot move from one phase to
another not using the new edges.

%!TEX root = ../../main.tex

\begin{figure}[t]
\center
\begin{minipage}[t]{.96\textwidth}
\begin{minipage}[t]{.31\textwidth}
	\begin{lstlisting}[style=condensed]
// Initialization		
while (true) {
	if (*) goto L1;
	Node top=*;
	if (top==NULL){
		STOP;
	}
	Node next=top.next;
	if(CAS(*,top,next)){
		out=top.data;
		STOP;
	}
}
STOP;
	\end{lstlisting}
\end{minipage}
\hfill
\begin{minipage}[t]{.33\textwidth}
	\begin{lstlisting}[style=condensed]
// CAS block
L1:Node top=ToS; 
goto L2;
while (true) {
	STOP;
	L2:if(top==NULL){
		STOP;
	}
	Node next=top.next;
	if(*) goto L3;
	if(CAS(ToS,top,next)){
		out=top.data;
		STOP;
	}
}
STOP;
L3:if(CAS(ToS,top,next)) 
   goto L4;
	\end{lstlisting}
\end{minipage}
\hfill
\begin{minipage}[t]{.30\textwidth}
	\begin{lstlisting}[style=condensed]
// Finalization
while (true) {
	Node top=*;
	if(top==NULL){
		return;
	}
	Node next=top.next;
	if(CAS(*,top,next)){
		L4:out=top.data;
		return;
	}
}
	\end{lstlisting}
\end{minipage}
\end{minipage}
\vspace{-6mm}
\caption{Step 1 in the summary computation for the \code{pop} method in Treiber's stack.}
\label{Figure:Phase1PopTreiber}
\vspace{-1mm}
\end{figure}

% ------------------------------------------------------------------
% ------------------------------------------------------------------
% ------------------------------------------------------------------
% ------------------------------------------------------------------

\paragraph{Step 2: Cleaning-up and summary generation.}
We perform \emph{copy propagation} using a must analysis that propagates an
assignment \code{y=x} to subsequent assignments \code{z=y}, resulting in \code{z=x}. 
That it is a must analysis means the propagation is done only if \code{z=y} definitely has to
use the value of \code{y} that stems from the assignment \code{y=x}. Moreover, we perform
the copy propagation assuming that the entire summary executes atomically.
For the initialization phase, the result is that the non-deterministic values for shared
variables propagate through the code. Similarly, for the CAS block, the shared
variables themselves propagate through the code. For the finalization phase,
non-deterministic values propagate only in the case when a local variable does not receive
its value from the CAS block. As a result, after the copy propagation, the CAS
and the finalization block may contain conditionals that are constantly true or
constantly false. We replace those that evaluate to true by skip and remove the
edges that evaluate to false. The result of the copy propagation is illustrated
in \cref{Figure:CopyPropagationPopTreiber}.

%!TEX root = ../../main.tex

\begin{figure}[t]
\vspace{-1mm}
\center
\begin{minipage}[t]{.96\textwidth}
\begin{minipage}[t]{.27\textwidth}
	\begin{lstlisting}[style=condensed]
// Initialization		
while (true) {
	if (*) goto L1;
	Node top=*;
	if(*){
		STOP;
	}
	Node next=*;
	if(*){
		out=*;
		STOP;
	}
}
STOP;
	\end{lstlisting}
\end{minipage}
\hfill
\begin{minipage}[t]{.38\textwidth}
	\begin{lstlisting}[style=condensed]
// CAS block
L1:Node top=ToS; 
   goto L2;
while (true) {
	STOP;
	L2:if(ToS==NULL){
		STOP;
	}
	Node next=ToS.next;
	if(*) goto L3;
	ToS=ToS.next;
	out=ToS.data;
	STOP;
}
STOP;
L3:if(CAS(ToS,ToS,ToS.next)) 
   goto L4;
	\end{lstlisting}
\end{minipage}
\hfill
\begin{minipage}[t]{.28\textwidth}
	\begin{lstlisting}[style=condensed]
// Finalization
while (true) {
	Node top=*;
	if(*){
		return;
	}
	Node next=*;
	if(*){
		L4:out=ToS.data;
		return;
	}
}
	\end{lstlisting}
\end{minipage}
\end{minipage}
\vspace{-6mm}
\caption{Copy propagation within the summary computation for \code{pop} in Treiber's stack.}
\label{Figure:CopyPropagationPopTreiber}
 \vspace{-4mm}
\end{figure}

Subsequently, we perform a \emph{live variables analysis}. A variable is live if
it may occur in a subsequent conditional or on the right-hand side of a
subsequent assignment.
Otherwise, it is dead.
We remove all assignments to dead variables including output parameters.
In our running example, all assignments to local variables as well as to the
output parameter can be removed. 

Next, we remove code that is \emph{unreachable}, \emph{dead}, or \emph{useless}.
Unreachable code can appear due to modifications of the CFG. Dead code does not
lead to the final location. Useless code does not have any impact on the values
of the variables used, which can concern even (possibly infinite) useless loops.

Finally, in the resulting code, conditionals are replaced by \code{assume} statements,
and the entire code is wrapped into an atomic block. For the pop method in
Treiber's stack, we get the summary given in \cref{code:treibers}.

	%!TEX root = ../main.tex

%==========================================================================
\section{Generalization to Explicit Memory Management}
\label{sec:generalization-to-explicit-memory-management}
%==========================================================================

%%%%%%%%%%%%%%%%%%%%%%%%%%%%%%%%%%%%%%%%%%%%%%%%%%%%%%%%%%%%%%%%%%%%%%%%%%%%%%
%%%%%%%%%%%%%%%%%%%%%%%%%%%%% INTRO + PRIMITIVES %%%%%%%%%%%%%%%%%%%%%%%%%%%%%
%%%%%%%%%%%%%%%%%%%%%%%%%%%%%%%%%%%%%%%%%%%%%%%%%%%%%%%%%%%%%%%%%%%%%%%%%%%%%%

We now generalize our approach to explicit memory management.
The problem is that the separation between the shared and owned heap is
difficult to define and establish in this case. Ownership as understood in
garbage collection, where no other thread can access a cell that was
allocated by another thread but not made shared, does not exist any more.
Memory can be freed and \emph{reallocated}, with other threads still holding
(dangling) pointers to it. These threads can read and modify that memory, hence
the allocating thread does not have strong guarantees of exclusivity. However,
programmers usually try to prevent effects of accidental reallocations: threads
are designed to \emph{respect~ownership}. That is, a thread should be allowed to
execute \emph{as if} it had exclusive access to the~memory~it~owns.

Our development is parameterized by a notion of separation between the shared and owned heap.
To generalize the results,  we provide a new notion of ownership suitable for explicit memory management.
However, the new notion is not guaranteed to be preserved by the semantics.
Instead, we include into our fixed-point computation a check that the program respects this ownership, and give up the analysis if the check fails.

%
% pointer primitives
%
To understand how the heap separation is influenced by basic pointer manipulations, we consider the following set of commands $\EBNFprimitive$:
\begin{align*}
	x = \cmdmalloc,\;
	x = \cmdfree,\;
	x = y,\;
	x = y\selsel_i,\;
	x\selsel_i = y\ .
\end{align*}
Here, $x, y$ are program variables and $\ptrsel_0,\dots,\ptrsel_n$ are selectors, from which the first, say $m$, are \emph{pointer selectors} and the rest are \emph{data selectors}.
Command $x = \cmdmalloc$ allocates a \emph{record}, a free block of addresses $a+0,\ldots,a+n$, and sets $h(x)$ to $a$.
Command $x = \cmdfree$ frees the record $h(x)+0,\ldots,h(x)+n$.
Selectors correspond to field accesses: $x\selsel_i$ refers to the content of $a+i$ if $x$ points to $a$.
The remaining commands have the expected meaning.

%%%%%%%%%%%%%%%%%%%%%%%%%%%%%%%%%%%%%%%%%%%%%%%%%%%%%%%%%%%%%%%%%%%%%%%%%%%%%%
%%%%%%%%%%%%%%%%%%%%%%%%%%%%%%%%% SEPARATION %%%%%%%%%%%%%%%%%%%%%%%%%%%%%%%%%
%%%%%%%%%%%%%%%%%%%%%%%%%%%%%%%%%%%%%%%%%%%%%%%%%%%%%%%%%%%%%%%%%%%%%%%%%%%%%%

\vspace{-2mm}
\subsection{Heap Separation}
\vspace{-1mm}

%
% Separation (free, shared, owned)
%
We work with a three-way partitioning of the heap into \emph{shared}, \emph{owned}, and \emph{free} addresses.
Free are all addresses that are fresh or have been freed and not reallocated.
Shared is every address that is \emph{reachable} from the shared variables and not free.
The reachability predicate, however, requires care.
% First, we must obviously generalize reachability from cells to whole records.
First, we must obviously generalize reachability from the first memory cell of a record to the whole record. 
Second, we must not use \emph{undefined} pointers for reachability.
% Values that were propagated from uninitialized or uncontrolled memory lead to \emph{undefined} pointers and variables.
A pointer is undefined if it was propagated from uninitialized or uncontrolled memory.
Letting the shared heap propagate through such values would make it possible for the entire allocated heap to be shared (since undefined pointers can have an arbitrary value).
Then, owned is the memory which is not shared nor free.
The owned memory is partitioned into disjoint blocks that are \emph{owned by individual threads}.
A thread gains ownership by moving memory into the owned part, and loses it when the memory is removed from the owned part.
The actions by which a thread can gain ownership are
\begin{inparaenum}[(1)]
 	\item allocation and
 	\item breaking reachability from shared variables by an update of a pointer or a shared variable (\emph{ownership transfer}).
\end{inparaenum}
An \emph{ownership violation} is then a modification of a thread's owned memory by another thread.
This can in particular be
\begin{inparaenum}[(1)]
	\item freeing or publishing the owned memory or
	\item an update of a~pointer therein.
\end{inparaenum}
A program respects ownership if it cannot reach an ownership violation.

%
% Formal definition
%
Let us discuss these concepts formally.
We use $\bot$ to identify free cells.
That is, in a heap $h$ address $a$ is free if $h(a)=\bot$ (also written $a\mkern-2mu\notin\mkern-3mu\dom{h}$).
A record is free if so are all its cells.
Consequently, the \code{free} command sets all cells of a record to~$\bot$.
The shared heap is identified by reachability through defined pointers starting from the shared variables. 
For undefined pointers we use the symbolic value $\udef$. 
Initially, all variables are undefined.
Moreover, we let allocations initialize the selectors of records to $\udef$.
We use a value distinct from $\bot$ to detect ownership violations by checking whether $\bot$ is reachable from the shared heap (see below). 
Value $\udef$ is explicitly allowed to be reachable (this may be needed for list implementations where selectors of sentinel nodes are not initialized).
Let $\Ptrshared$ be the shared pointer variables.
Then, the addresses of the shared records in a heap $h$, denoted by $\Recshared{h}\subseteq\Nat\cup\set{\bot}$, are collected by the following fixed point (where the \emph{address of a record} is its lowest address):
\begin{align*}
	S_0 &= \setcond{a}{\exists x\in\Ptrshared.\, h(x)=a\neq\udef}
	\\
	S_{i+1} &= \setcond{b}{\exists\, a{\,\in\,} S_i\: \exists\, k.~ a\neq\bot \:\wedge\: 0\leq k < m \:\wedge\: h(a+k) =b\neq\udef}
\end{align*}
All addresses within the shared records are then shared.
The remaining cells, \ie, those that are neither free nor shared, are owned.
This definition establishes a sufficient separation for \cref{assumption:sequential-separation}.
It is automatically lifted to the concurrent setting by Rule~(\textsc{par}) following the intuition from above.

%
% Detecting ownership violations
%
It remains to detect ownership violations, which occur whenever a thread modifies cells owned by other threads.
Due to the separation integrated into Rule~(\textsc{par}), threads execute with only the shared heap and their owned heap being visible.
The remainder of the heap is cut away.
By choice of $\bot$ to identify free cells, the cut away part appears free to the acting thread.
In particular, the parts owned by other threads appear free.
Hence, in order to avoid ownership violations, a thread must not modify free cells.
To that end, an ownership violation occurs if
\begin{inparaenum}[(A)]
	\item
		\label{new:generalization:ovio:free}
		a free cell is freed again,
	\item
		\label{new:generalization:ovio:write}
		a free cell is written to, or
	\item
		\label{new:generalization:ovio:publish}
		a free cell is published to the shared heap.
\end{inparaenum}
For \eqref{new:generalization:ovio:free} and \eqref{new:generalization:ovio:write} we extend the semantics of commands to raise an ownership violation if a free cell is manipulated.
For \eqref{new:generalization:ovio:publish} we analyse all program steps of threads and check whether the step results in a shared heap where $\bot$ is made reachable.

Formally, we have the following rules.
\begin{gather*}
	% Rule: free owned
 	\rname{\ref{new:generalization:ovio:free}}
	\raxiom{
		\exists\ptrsel.\, (s\uplus o)(x)\selsel\notin\dom{s\uplus o}
	}
	\rconclusion{1}{
		\state{s}{\conf{x=\cmdfree}{o}}
		\stepprog
		\mathit{violation}
	}
	\DisplayProof
	\qquad
	%
	% Rule: modify owned
 	\rname{\ref{new:generalization:ovio:write}}
	\raxiom{
		(s\uplus o)(x)\selsel\notin\dom{s\uplus o}
	}
	\rconclusion{1}{
		\state{s}{\conf{x\selsel=y}{o}}
		\stepprog
		\mathit{violation}
	}
	\DisplayProof
	\\[2mm]
	%
	% Rule: publish owned
 	\rname{\ref{new:generalization:ovio:publish}}
	\raxiom{
		\state{s}{\cf}\stepprog\state{s'}{\cf'}
		\qquad
		\bot\in\Recshared{s'}
	}
	\rconclusion{1}{
		\state{s}{\cf}
		\stepprog
		\mathit{violation}
	}
	\DisplayProof
\end{gather*}
Note that reading out free cells is allowed by the above rules.
This is necessary because lock-free algorithms typically perform speculating reads and check only later whether the result of the read is safe to use.
% This is often used in lock-free algorithms to check (with a CAS) whether a local copy of a shared variable is still up-to-date. 
Moreover, note that our detection of ownership violations can yield false-positives.
A cell may not be owned, yet an ownership violation is raised because it appears free to the thread.
We argue that such false-positives are \emph{desired} as they access truly free memory.
Put differently: an ownership violation detected by the above rules is either indeed an ownership violation or an unsafe access of free memory, that is, a bug.

%%%%%%%%%%%%%%%%%%%%%%%%%%%%%%%%%%%%%%%%%%%%%%%%%%%%%%%%%%%%%%%%%%%%%%%%%%%%%%
%%%%%%%%%%%%%%%%%%%%%%%%%%%%% OWNERSHIP TRANSFER %%%%%%%%%%%%%%%%%%%%%%%%%%%%%
%%%%%%%%%%%%%%%%%%%%%%%%%%%%%%%%%%%%%%%%%%%%%%%%%%%%%%%%%%%%%%%%%%%%%%%%%%%%%%

\vspace{-2mm}
\subsection{Ownership Transfer}
\vspace{-1mm}

%
% Ownership transfer
%
The above separation is different from the one used under garbage collection in the earlier sections. 
When an address becomes unreachable from the shared variables, 
it is \emph{transferred} into the acting thread's owned heap 
(although other threads may still have pointers to it).
We introduce this ownership transfer to simplify the construction of summaries.
The idea is best understood on an example. 

%
% Ownership transfer example
%
Under explicit memory management, threads free cells that they made unreachable from the shared variables to avoid memory leaks. 
Consider, for example, the method \texttt{pop} in Treiber's stack (\cref{code:treibers}). 
There, a thread updates the \texttt{ToS} variable making the former top of stack, say $a$, unreachable from the shared heap. 
In the version for explicit memory management, $a$ is then freed before returning. 
If ownership was not transferred and address $a$ stayed shared, then two summaries would be needed: one for the update of \texttt{ToS} and one for freeing~$a$. 
However, a stateless version of the latter summary could not learn which address to free since it starts with the empty local heap and with $a$ unreachable from the shared heap. 
If, on the other hand, ownership of $a$ is transferred to the acting thread, then the former summary can include freeing $a$ (which does not change the shared heap).  
Moreover, it is even forced to free $a$ in order to remain stateless since $a$ would otherwise persist in its owned heap.\\[0.2cm]

%
% Note on other ownership definitions
%
\vspace{-6mm}
We stress that our framework can be instantiated with other notions of separation,
like an analogue of the one for garbage collection or the one of \cite{DBLP:conf/vmcai/HazizaHMW16}, 
which both do not have ownership transfer. 
This would complicate the reasoning in \cref{sec:effect-summaries},
but could lead to a more robust analysis (ownership transfer is prone to ownership~violations).

%%%%%%%%%%%%%%%%%%%%%%%%%%%%%%%%%%%%%%%%%%%%%%%%%%%%%%%%%%%%%%%%%%%%%%%%%%%%%%
%%%%%%%%%%%%%%%%%%%%%%%%%%%%%% VERSION COUNTERS %%%%%%%%%%%%%%%%%%%%%%%%%%%%%%
%%%%%%%%%%%%%%%%%%%%%%%%%%%%%%%%%%%%%%%%%%%%%%%%%%%%%%%%%%%%%%%%%%%%%%%%%%%%%%

\vspace{-2mm}
\subsection{ABA Prevention}
\vspace{-1mm}

%
% ABA problem
%
Additional synchronization mechanisms can be incorporated into our approach.
For instance, lock-free data structures may use \emph{version counters} to prevent the \emph{ABA problem}~\cite{DBLP:journals/jpdc/MichaelS98}: 
a variable leaves and returns to the same address, and an observer incorrectly concludes that the variable has never changed.
A well-known scenario of this type causes stack corruption in a naive extension of Treiber's stack to explicit memory management~\cite{DBLP:journals/jpdc/MichaelS98}.
To give the observer a means of detecting that a variable has been changed, pointers are associated with a counter that increases with every update.

%
% Version counters in our framework
%
In our analysis, such version counters must be persistent in the shared memory.
Since this is an exception from the above definition of separation, a presence of version counters must be indicated by the user 
(e.g., the user specifies that the version counter of a pointer $a$ is always stored at address $a+1$).
The semantics is then adapted in such a way that
\begin{inparaenum}[(1)]
 	\item version counters remain in the shared heap upon freeing,
 	\item are retained in case of reallocations, and
 	\item are never transferred to a thread's owned heap.
\end{inparaenum}
The modifications can be easily implemented, and are detailed in Appendix~\ref{app:version_counters}.
Last, the thread-modular abstraction has to be adjusted since keeping all counters ever allocated in every thread view is not feasible. 
One solution is to remember only the values of those counters that are attached to the allocated shared and the thread's own heap.
%

	%!TEX root = ../main.tex

%==============================================================================
\section{Experiments and Discussion}
\label{sec:experiments}
%==============================================================================

To substantiate our claim for practical benefits of the proposed method, we implemented
the techniques from Sections~\cref{sec:effect-summaries}.\footnote{Available at: \url{\githuburl}}
Therefore, we modified our previous linearizability checker \cite{DBLP:conf/vmcai/HazizaHMW16} to perform our novel fixed-point computation.
The modifications were straightforward and implemented along the lines of \cref{sec:effect-summaries}.

%!TEX root = ../../main.tex

\begin{table}[t]
	\vspace{-1mm}
	% Hardware: Intel Xeon E5-2670@2.60GHz.
	\newcommand{\cell}[2]{%
		\begin{minipage}{1cm}
			\hfill#1
		\end{minipage} / 
		\begin{minipage}{1cm}
			#2
		\end{minipage}}
	\def\firstcolwidth{4.2cm}%
	\setlength\extrarowheight{1pt}%
	\newcolumntype{Y}{>{\centering\arraybackslash}X}%
	\newcolumntype{Z}{>{\raggedright}m}%
	\begin{center}
		\begin{tabularx}{1.0\textwidth}{Z{\firstcolwidth+.1cm}l*{2}{Y}}
			\multicolumn{2}{l}{Program}
				& Thread-Modular \cite{DBLP:conf/vmcai/HazizaHMW16}
				& Thread Summaries
				\\
			\hline\hline
			Coarse stack
				& GC & \cell{0.29s}{343} & \cell{0.03s}{256} \\
				& MM & \cell{1.89s}{1287} & \cell{0.19s}{1470} \\
			\hline
			Coarse queue
				& GC & \cell{0.49}{343} & \cell{0.05s}{256} \\
				& MM & \cell{2.34s}{1059} & \cell{0.98s}{2843} \\
			\hline
			Treiber's stack \cite{opac-b1015261}
				& GC & \cell{1.99s}{651}  & \cell{0.06s}{458} \\
				& MM & \cell{25.5s}{3175} & \cell{1.64s}{2926} \\
			\hline
			Michael\&Scott's queue \cite{DBLP:journals/jpdc/MichaelS98}
				& GC & \cell{11.0s}{1530} & \cell{0.39s}{1552} \\
				& MM & \cell{11700s}{19742}  & \cell{102s}{27087} \\
				% & MM & \cell{39907s}{26718}  & \cell{102s}{27087} \\
				%% & MM & \cell{196m}{19742}  & \cell{102s}{27087} \\
			\hline
			DGLM queue \cite{DBLP:conf/forte/DohertyGLM04}
				& GC & \cell{9.56s}{1537}  & \cell{0.37s}{1559} \\
				& MM & unsafe (spurious)  & violation \\
			\hline
		\end{tabularx}
	\end{center}
	\caption{Experimental results: a speed-up of up to two orders of magnitude.}%
	\label{table:experiments}%
	\vspace{-6mm}
\end{table}

Our findings are listed in \cref{table:experiments}.
Experiments were conducted on an Intel Xeon E5-2670 running at 2.60GHz.
The table includes running times (averaged over ten runs) and the number of explored views (size of the computed fixed point).
Our benchmarks include well-known data structures such as Treiber's lock-free stack \cite{opac-b1015261},  Michael\&Scott's lock-free queue \cite{DBLP:journals/jpdc/MichaelS98}, and the lock-free DGLM queue~\cite{DBLP:conf/forte/DohertyGLM04}.
We do not include lock-free set implementations due to limitations of the tool
in handling data---not due to limitations of our approach.
We ran~each benchmark under garbage collection~(GC) and explicit memory management~(MM).
Additionally, we include for each benchmark a comparison between our novel fixed point using summaries and the optimized version of the classical thread-modular fixed point from \cite{DBLP:conf/vmcai/HazizaHMW16}.

Our experiments show that relying on summaries provides a significant performance boost compared to classical interference.
This holds true for both garbage collection and explicit memory management.
For garbage collection, we experience a speed-up of one order of magnitude throughout the entire test suite.
Although comparisons among different implementations are inherently unfair, we note that our tool compares favorably to competitors \cite{DBLP:conf/tacas/AbdullaHHJR13,DBLP:conf/sas/AbdullaJT16,DBLP:conf/cav/Vafeiadis10,DBLP:conf/vmcai/Vafeiadis10}.
Under explicit memory management, the same speed-up is present for simple algorithms, like Treiber's stack.
For slightly more complex implementations, like Michael\&Scott's queue, we observe a more eminent speed-up of over two orders of magnitude.
This speed-up is present even though the analysis explores a~way larger search space than its classical counterpart.
This confirms that our approach of reducing the complexity of interference steps rather than reducing the search space is beneficial for verification.

Unfortunately, we could not establish correctness of the DGLM queue under explicit memory management with neither of the fixed points.
For the classical one, the reason was imprecision in the underlying shape analysis which resulted in spurious unsafe memory accesses.
For our novel fixed point, the tool detected an ownership violation according to \cref{sec:generalization-to-explicit-memory-management}.
While being correct, the DGLM queue indeed features such a violation.
The update pattern in the \texttt{deque} method can result in freeing nodes that were made unreachable by other threads.
The problematic scenario only occurs when the head of the queue overtakes the tail.
Despite the similarity, this behavior is not present in Michael\&Scott's queue which is why it does not suffer from such a violation.

As hinted in \cref{sec:generalization-to-explicit-memory-management}, one could
generalize our theory in such a way that no ownership transfer is required.
Without ownership transfer, however, freeing cells becomes an effect of the
shared heap which cannot be mimicked: a stateless summary cannot acquire a
pointer to an unreachable cell and thus not mimic the free. Consequently, one
has to relax % or even completely remove 
the assumption of statelessness. %This does not only complicate the theory. It
% also inflicts changes on the fixed point from \cref{sec:effect-summaries}.
This inflicts major changes on the fixed point from \cref{sec:effect-summaries}.
Besides program threads, it would need to include threads executing stateful
summaries. Moreover, one would need to reintroduce interference steps. However,
we argue that only such interference steps are required where stateful summaries
appear as the interfering thread.  Hence, the number of interference steps is expected
to be significantly lower than for ordinary interference.
We consider a proper investigation of these issues an interesting subject for
future~work.

	%!TEX root = ../main.tex

%==============================================================================
\vspace{-3mm}
\section{Related Work}
\label{sec:related_work}
\vspace{-2mm}
%==============================================================================

%% Thread-modular
We already commented on the two approaches of computing interference steps.
The merge-and-project approach
\cite{DBLP:conf/cav/BerdineLMRS08,DBLP:conf/spin/FlanaganQ03,DBLP:journals/toplas/Jones83,DBLP:conf/ictac/MalkisPR06}
suffers form low scalability and precision due to computing too many merge-compatible heaps.
To improve precision of interference, 
works like \cite{DBLP:conf/pldi/GotsmanBCS07,DBLP:conf/aplas/SegalovLMGS09,DBLP:conf/vmcai/Vafeiadis10} track additional thread correlations; ownership, for instance.
However, keeping more information within thread states usually has a negative impact on scalability.
Moreover, for the programs of our interest, those techniques were not applicable in the case of explicitly managed memory
which does not provide exclusivity guarantees. 
% since it does not provide exclusivity guarantees. 
%
Instead, \cite{DBLP:conf/cav/BerdineLMRS08,DBLP:conf/tacas/AbdullaHHJR13}
proposed to maintain views of two threads, allowing one to infer the context in which the views occur.
Since this again jeopardizes scalability, \cite{DBLP:conf/vmcai/HazizaHMW16} tailored ownership towards explicit memory management.
Still, computing interference remained quadratic in the size of the fixed point.
% and prone to relate local states of joined threads falsely.
%
Our approach improves dramatically on the efficiency of \cite{DBLP:conf/vmcai/HazizaHMW16} while keeping its precision. 

%
%The complexity of interference via summaries is linear in the size of the search space.
%Moreover, it does not need to perform potentially expensive compatibility checks.
%Precision-wise, statelessness avoids any need to correlate local states of threads.

The learning approach in \cite{DBLP:conf/vmcai/Vafeiadis09,DBLP:conf/vmcai/Vafeiadis10,DBLP:conf/concur/VafeiadisP07} and \cite{DBLP:conf/esop/Mine11,DBLP:conf/vmcai/Mine14,DBLP:conf/vmcai/MonatM17} performs a variant of rely/guarantee reasoning \cite{DBLP:conf/ifip/Jones83} paired with symbolic execution and abstract interpretation, respectively.
In a fixed point, the interference produced by a thread is recorded and applied to other threads in consecutive iterations.
This computes a symbolic representation of the inteference which is as precise as the underlying abstract domain   
(although the precision may be relaxed by further abstraction and hand-crafted joins).
Our method improves on this in various aspects.
First, we never compute the most precise interference information.
Our summaries can be understood as a form of interpolant between the most precise approximation and the complement of the bad states.
Second, our summaries are syntactic objects (program code) which are independent of the actual verification procedure and thus reusable.
The learned interference may be reused only in the same abstract domain it was computed in.
Third, we show how to lift our approach to explicit memory management what has not been done before.
Fourth, our results are independent of the actual program semantics relying only on a small core language.
Our development required to formulate the principles that libraries rely on (statelessness) which have not been made explicit elsewhere.

Another approach to make the verification of low-level implementations tractable is atomicity abstraction \cite{DBLP:conf/lics/AbadiL88,DBLP:conf/popl/ElmasQT09,DBLP:conf/tacas/ElmasQSST10,DBLP:journals/fac/Jonsson12,DBLP:conf/ecoop/PintoDG14,DBLP:conf/fmcad/PopeeaRW14}.
The core idea is to translate a given program into its specification by introducing and enlarging atomic blocks. The code transformations must be provably sound, with the soundness arguments oftentimes crafted for a particular semantics only.
While generating summaries is closely related to making the program under scrutiny more atomic, we pursue a different approach.
Our rewriting rules (\ie the computation of summaries) do not need to be, and indeed are not, provably sound, which allows for much more freedom.
Nevertheless, we guarantee a sound analysis.
Our sanity checks can be understood as an efficient, fully automatic procedure to check whether or not the applied atomicity abstraction was sound.
Additionally, we do not rely on a particular memory semantics.

Simulation relations are widely used for linearizability proofs \cite{DBLP:conf/forte/DohertyGLM04,DBLP:conf/tacas/ElmasQSST10,DBLP:journals/tocl/SchellhornDW14,DBLP:conf/ssiri/ZhangL10} and verified compilation \cite{DBLP:journals/toplas/JagannathanLPPV14,DBLP:journals/jar/Leroy09}.
There, one establishes a simulation relation between a low-level program and a high-level program stating that the latter preserves the behaviors of the former.
Verifying properties of the low-level program then reduces to verifying the same property for the high-level program.
Establishing simulation relations, however, suffers from the same shortcomings as atomicity~abstraction.

Finally, \cite{DBLP:conf/esop/GotsmanRY13} introduces \emph{grace periods}, an idiom similar to CAS blocks.
It reflects the protocol used by a program to prohibit data corruption.
During a grace period, it is guaranteed that a thread's memory is not freed.
However, no method for checking conformance to such periods is given.
That is, soundness cannot be checked when relying on grace periods whereas our sanity checks can efficiently detect unsound~verification~results.

	%!TEX root = ../main.tex

%==============================================================================
\vspace{-3mm}
\section{Conclusion}
\label{sec:conclusion}
\vspace{-2mm}
%==============================================================================

We proposed a new approach for verifying lock-free data structures.
The approach builds on the so-called CAS blocks (or, more generally,
copy-and-check code blocks) which are commonly used when implementing lock-free
data structures. We proposed a heuristic that builds stateless program summaries
from such blocks. By avoiding many expensive merge-and-project operations, the
approach can greatly increase the efficiency of thread-modular verification. This
was confirmed by our experimental results showing that the implementation of
our approach compares favorably with other competing tools. Moreover, our
approach naturally combines with recently proposed reasoning about ownership to
improve the precision of thread-modular reasoning, which allowed us to handle complex
lock-free code efficiently even under explicit memory management. Of course, our
heuristically computed stateless summaries can miss some reachable shared heaps,
but, as a major part of our contribution, we proved that one can check whether this
is the case on the generated state space. Hence, we can perform sound
verification using a potentially unsound abstraction.

In the future, we would like to investigate CEGAR to include missing effects into our summaries.
The main question here is how to refine the program code of a summary using an abstract representation of the missing effects. 
Further, it may be necessary to introduce stateful summaries in order to include certain effects, as revealed by the DGLM queue under explicit memory management.
Moreover, in theory, our approach could increase not only efficiency but also
precision compared with other approaches. This is due to the atomicity of the
CAS blocks that could rule out interleavings that other approaches would
explore.
We have not found this confirmed in our experiments.
Nevertheless, we find it worth investigating the theoretical and practical aspects of this matter in the future.

	\bibliographystyle{splncs03}
	% \bibliography{bib}
	\bibliography{bib_short}
	% \bibliography{bib_short_with_url}
	
	%!TEX root = ../../main.tex

\appendix

%!TEX root = ../../main.tex

\section{Data Structure Implementations and Summaries}

To complement our benchmarks from \cref{sec:experiments} we provide the considered data structure implementations and the corresponding summaries.
In the code, we mark with bold font the extensions needed for explicit memory management, but not for garbage collection.
For explicit memory management we use version counters to avoid the ABA problem.

\Cref{appendix:examples:common} gives an overview of structures and functions common to all algorithms.
The coarse stack can be found in \cref{appendix:examples:coarsestack:code}, its summary in \cref{appendix:examples:coarsestack:summary}.
The coarse queue can be found in \cref{appendix:examples:coarsequeue:code}, its summary in \cref{appendix:examples:coarsequeue:summary}.
Treiber's stack can be found in \cref{appendix:examples:treiber:code}, its summary in \cref{appendix:examples:treiber:summary}.
Michael\&Scotts's queue can be found in \cref{appendix:examples:ms:code}, its summary in \cref{appendix:examples:ms:summary}.
The DGLM queue can be found in \cref{appendix:examples:dglm:code}, its summary in \cref{appendix:examples:dglm:summary}.
Note that the DGLM queue does not have a stateless summary under explicit memory management as discussed in \cref{sec:experiments}.

%%%%%%%%%%%%%%%%%%%%%%%%%%%%%%%%%%%%%%%%%%%%%%%%%%%%%%%%%%%%%%%%%%%%%%%%%%%%%%
%%%%%%%%%%%%%%%%%%%%%%%%%%%%%%%%%%% COMMON %%%%%%%%%%%%%%%%%%%%%%%%%%%%%%%%%%%
%%%%%%%%%%%%%%%%%%%%%%%%%%%%%%%%%%%%%%%%%%%%%%%%%%%%%%%%%%%%%%%%%%%%%%%%%%%%%%

\begin{figure}
	\center
	\begin{minipage}{.92\textwidth}
		\begin{minipage}[t]{.42\textwidth}
			\begin{lstlisting}[style=condensed]
struct ptr_t {
	@uint  age;@
	Node* ptr;
}

struct Node {
	data_t data;
	ptr_t  next;
}
				\end{lstlisting}
			\end{minipage}
			\hfill
			\begin{minipage}[t]{.5\textwidth}
				\begin{lstlisting}[style=condensed]
ptr_t create(data_t data) {
	ptr_t res;
	res.ptr = new Node();
	res.ptr->data = data;
	res.ptr->next.ptr = NULL;
	@// version counters are not
	// initialized intentionally
	@return res;
}
			\end{lstlisting}
		\end{minipage}
		\vspace{-3mm}
	\end{minipage}
	\begin{minipage}{.92\textwidth}
		\begin{lstlisting}[style=condensed]
atomic boolean CAS(ptr_t& dst, ptr_t cmp, ptr_t swp) {
	if (@dst.age == cmp.age &&@ dst.ptr == cmp.ptr) {
		@dst.age = cmp.age + 1;@
		dst.ptr = swp.ptr;
		return true;
	} else {
		return false;
	}
}
		\end{lstlisting}
	\end{minipage}
	\vspace{-3mm}
	\caption{Common functionality.}
	\label{appendix:examples:common}
	\vspace{-7mm}
\end{figure}

%%%%%%%%%%%%%%%%%%%%%%%%%%%%%%%%%%%%%%%%%%%%%%%%%%%%%%%%%%%%%%%%%%%%%%%%%%%%%%
%%%%%%%%%%%%%%%%%%%%%%%%%%%%%%%% COARSE STACK %%%%%%%%%%%%%%%%%%%%%%%%%%%%%%%%
%%%%%%%%%%%%%%%%%%%%%%%%%%%%%%%%%%%%%%%%%%%%%%%%%%%%%%%%%%%%%%%%%%%%%%%%%%%%%%

\begin{figure}
	\begin{minipage}{.92\textwidth}
		\center
		\begin{minipage}[t]{.52\textwidth}
			\begin{lstlisting}[style=condensed]
shared ptr_t ToS;

void init() {
	ToS.ptr = NULL;
}

void push(data_t in) {
	ptr_t node = create(in);
	Atomic {
		node.ptr->next = ToS;
		ToS = node;
	}
	return;
}
				\end{lstlisting}
			\end{minipage}
			\hfill
			\begin{minipage}[t]{.4\textwidth}
				\begin{lstlisting}[style=condensed]
bool pop(data_t& out) {
	Atomic {
		if (ToS.ptr == NULL) {
			return false;
		} else {
			out = ToS.ptr->data;
			@ptr_t top = ToS;@
			ToS = ToS.ptr->next;
			@free(top.ptr);@
			return true;
		}
	}
}
			\end{lstlisting}
		\end{minipage}
	\end{minipage}
	\vspace{-4mm}
	\caption{Coarse stack.}
	\label{appendix:examples:coarsestack:code}
\end{figure}

\begin{figure}
	\begin{minipage}{.92\textwidth}
		\center
		\begin{minipage}[t]{.52\textwidth}
			\begin{lstlisting}[style=condensed]
S1: atomic {
	/* push */
	ptr_t node = create(*);
	node.ptr->next = ToS;
	ToS = node;
}

S2: atomic { /* skip */ }
				\end{lstlisting}
			\end{minipage}
			\hfill
			\begin{minipage}[t]{.4\textwidth}
				\begin{lstlisting}[style=condensed]
S3: atomic {
	/* pop */
	assume(ToS.ptr != NULL);
	@ptr_t old = ToS;@
	ToS.ptr = ToS.ptr->next.ptr;
	@free(old.ptr);@
}
			\end{lstlisting}
		\end{minipage}
	\end{minipage}
	\vspace{-4mm}
	\caption{Summaries for coarse stack.}
	\label{appendix:examples:coarsestack:summary}
	% \vspace{-7mm}
\end{figure}

%%%%%%%%%%%%%%%%%%%%%%%%%%%%%%%%%%%%%%%%%%%%%%%%%%%%%%%%%%%%%%%%%%%%%%%%%%%%%%
%%%%%%%%%%%%%%%%%%%%%%%%%%%%%%%% COARSE QUEUE %%%%%%%%%%%%%%%%%%%%%%%%%%%%%%%%
%%%%%%%%%%%%%%%%%%%%%%%%%%%%%%%%%%%%%%%%%%%%%%%%%%%%%%%%%%%%%%%%%%%%%%%%%%%%%%

\begin{figure}
	\begin{minipage}{.92\textwidth}
		\center
		\begin{minipage}[t]{.52\textwidth}
			\begin{lstlisting}[style=condensed]
shared ptr_t Head, Tail;

void init() {
	Head = create(*);
	Tail = Head;
}

void enq(data_t in) {
	ptr_t node = create(in);
	Atomic {
		Tail.ptr->next = node;
		Tail = node;
	}
	return;
}
				\end{lstlisting}
			\end{minipage}
			\hfill
			\begin{minipage}[t]{.4\textwidth}
				\begin{lstlisting}[style=condensed]
bool pop(data_t& out) {
	Atomic {
		if (Head.ptr->next == NULL) {
			return false;
		} else {
			out = Head.ptr->next.ptr->data;
			@ptr_t head = Head;@
			Head = Head.ptr->next;
			@free(head.ptr);@
			return true;
		}
	}
}
			\end{lstlisting}
		\end{minipage}
	\end{minipage}
	\vspace{-6mm}
	\caption{Coarse queue.}
	\label{appendix:examples:coarsequeue:code}
\end{figure}

\begin{figure}
	\begin{minipage}{.92\textwidth}
		\center
		\begin{minipage}[t]{.52\textwidth}
			\begin{lstlisting}[style=condensed]
S1: atomic {
	/* push */
	ptr_t node = create(*);
	Tail.ptr->next = node;
	Tail = node;
}

S2: atomic { /* skip */ }
				\end{lstlisting}
			\end{minipage}
			\hfill
			\begin{minipage}[t]{.4\textwidth}
				\begin{lstlisting}[style=condensed]
S3: atomic {
	/* pop */
	assume(Head.ptr->next != NULL);
	@ptr_t old = Head;@
	Head = Head.ptr->next;
	@free(old.ptr);@
}
			\end{lstlisting}
		\end{minipage}
	\end{minipage}
	\vspace{-4mm}
	\caption{Summaries for coarse queue.}
	\label{appendix:examples:coarsequeue:summary}
\end{figure}

%%%%%%%%%%%%%%%%%%%%%%%%%%%%%%%%%%%%%%%%%%%%%%%%%%%%%%%%%%%%%%%%%%%%%%%%%%%%%%
%%%%%%%%%%%%%%%%%%%%%%%%%%%%%%%%%% TREIBER %%%%%%%%%%%%%%%%%%%%%%%%%%%%%%%%%%%
%%%%%%%%%%%%%%%%%%%%%%%%%%%%%%%%%%%%%%%%%%%%%%%%%%%%%%%%%%%%%%%%%%%%%%%%%%%%%%

\begin{figure}
	\begin{minipage}{.92\textwidth}
		\center
		\begin{minipage}[t]{.52\textwidth}
			\begin{lstlisting}[style=condensed]
shared ptr_t ToS;

void init() {
	ToS.ptr = NULL;
}

void push(data_t in) {
	ptr_t node = create(in);
	while (true) {
		ptr_t top = ToS;
		node.ptr->next = top;
		if(CAS(ToS, top, node)){
			return;
		}
	}
}
				\end{lstlisting}
			\end{minipage}
			\hfill
			\begin{minipage}[t]{.4\textwidth}
				\begin{lstlisting}[style=condensed]
bool pop(data_t& out) {
	while (true) {
		ptr_t top = ToS;
		if(top.ptr == NULL){
			return false;
		}
		ptr_t next = top.ptr->next;
		if(CAS(ToS, top, next)){
			out = top.ptr->data;
			@free(top.ptr);@
			return true;
		}
	}
}
			\end{lstlisting}
		\end{minipage}
	\end{minipage}
	\vspace{-6mm}
	\caption{Treiber's lock-free stack \cite{opac-b1015261}.}
	\label{appendix:examples:treiber:code}
\end{figure}

\begin{figure}
	\begin{minipage}{.92\textwidth}
		\center
		\begin{minipage}[t]{.52\textwidth}
			\begin{lstlisting}[style=condensed]
S1: atomic {
	/* push */
	ptr_t node = create(*);
	node.ptr->next = ToS;
	@node.age = ToS.age + 1;@
	ToS = node;
}

S2: atomic { /* skip */ }
				\end{lstlisting}
			\end{minipage}
			\hfill
			\begin{minipage}[t]{.4\textwidth}
				\begin{lstlisting}[style=condensed]
S3: atomic {
	/* pop */
	assume(ToS.ptr != NULL);
	@ptr_t old = ToS;@
	ToS.ptr = ToS.ptr->next.ptr;
	@ToS.age++;@
	@free(old.ptr);@
}
			\end{lstlisting}
		\end{minipage}
	\end{minipage}
	\vspace{-4mm}
	\caption{Summaries for Treiber's lock-free stack.}
	\label{appendix:examples:treiber:summary}
\end{figure}

%%%%%%%%%%%%%%%%%%%%%%%%%%%%%%%%%%%%%%%%%%%%%%%%%%%%%%%%%%%%%%%%%%%%%%%%%%%%%%
%%%%%%%%%%%%%%%%%%%%%%%%%%%%%%% MICHAEL&SCOTT %%%%%%%%%%%%%%%%%%%%%%%%%%%%%%%%
%%%%%%%%%%%%%%%%%%%%%%%%%%%%%%%%%%%%%%%%%%%%%%%%%%%%%%%%%%%%%%%%%%%%%%%%%%%%%%

\begin{figure}
	\begin{minipage}{.92\textwidth}
		\center
		\begin{minipage}[t]{.52\textwidth}
			\begin{lstlisting}[style=condensed]
shared ptr_t Head, Tail;


void init() {
	Head = create(*);
	Tail = Head;
}

void enq(data_t in) {
	ptr_t node = create(in);
	while (true) {
		ptr_t tail = Tail;
		ptr_t next = tail.ptr->next;
		@if (tail.age != Tail.age) continue;@
		if (tail.ptr != Tail.ptr) continue;
		if (next.ptr == NULL) {
			if (CAS(tail.next, next, node)) {
				CAS(Tail, tail, next);
				return;
			}
		} else {
			CAS(Tail, tail, next);
		}
	}
}
			\end{lstlisting}
		\end{minipage}
		\hfill
		\begin{minipage}[t]{.4\textwidth}
			\begin{lstlisting}[style=condensed]
bool deq(data_t& out) {
	while (true) {
		ptr_t head = Head;
		ptr_t tail = Tail;
		ptr_t next = head.ptr->next;
		@if (head.age != Head.age) continue;@
		if (head.ptr != Head.ptr) continue;
		if (head.ptr == tail.ptr) {
			if (next.ptr == NULL) return false;
			else CAS(Tail, tail, next);
		} else {
			out = next.ptr->data;
			if (CAS(Head, head, next)) {
				@free(head.ptr);@
				return true;
			}
		}
	}
}
			\end{lstlisting}
		\end{minipage}
	\end{minipage}
	\vspace{-6mm}
	\caption{Michael\&Scott's non-blocking queue \cite{DBLP:journals/jpdc/MichaelS98}.}
	\label{appendix:examples:ms:code}
\end{figure}

\begin{figure}
	\begin{minipage}{.92\textwidth}
		\center
		\begin{minipage}[t]{.52\textwidth}
			\begin{lstlisting}[style=condensed]
S1: atomic {
	/* enq (add new node) */
	assume(Tail.ptr->next.ptr == NULL);
	ptr_t node = create(*);
	Tail.ptr->next.ptr = node.ptr;
	Tail.ptr->next.age++;
	@Tail.age++;@
}

S2: atomic {
	/* enq (swing Tail) */
	assume(Tail.ptr->next.ptr != NULL);
	Tail.ptr = Tail.ptr->next.ptr;
	@Tail.age++;@
}

S3: atomic { /* skip */ }
				\end{lstlisting}
			\end{minipage}
			\hfill
			\begin{minipage}[t]{.4\textwidth}
				\begin{lstlisting}[style=condensed]
S4: atomic {
	/* deq (swing Tail) */
	assume(Head.ptr == Tail.ptr);
	assume(Head.ptr->next.ptr != NULL);
	Tail.ptr = Tail.ptr->next.ptr;
	@Tail.age++;@
}

S5: atomic {
	/* deq */
	assume(Head.ptr != Tail.ptr);
	@ptr_t old = Head;@
	Head.ptr = Head.ptr->next.ptr;
	@Head.age++;@
	@free(old.ptr);@
}
			\end{lstlisting}
		\end{minipage}
	\end{minipage}
	\vspace{-4mm}
	\caption{Summary for Michael\&Scott's non-blocking queue.}
	\label{appendix:examples:ms:summary}
\end{figure}

%%%%%%%%%%%%%%%%%%%%%%%%%%%%%%%%%%%%%%%%%%%%%%%%%%%%%%%%%%%%%%%%%%%%%%%%%%%%%%
%%%%%%%%%%%%%%%%%%%%%%%%%%%%%%%%%%%% DGLM %%%%%%%%%%%%%%%%%%%%%%%%%%%%%%%%%%%%
%%%%%%%%%%%%%%%%%%%%%%%%%%%%%%%%%%%%%%%%%%%%%%%%%%%%%%%%%%%%%%%%%%%%%%%%%%%%%%

\begin{figure}
	\begin{minipage}{.92\textwidth}
		\center
		\begin{minipage}{.52\textwidth}
			\begin{lstlisting}[style=condensed]
shared ptr_t Head, Tail;


void init() {
	Head = create(*);
	Tail = Head;
}

void enq(data_t in) {
	ptr_t node = create(in);
	while (true) {
		ptr_t tail = Tail;
		ptr_t next = tail.ptr->next;
		@if (tail.age != Tail.age) continue;@
		if (tail.ptr != Tail.ptr) continue;
		if (next.ptr == NULL) {
			if (CAS(tail.next, next, node)) {
				CAS(Tail, tail, next);
				return;
			}
		} else {
			CAS(Tail, tail, next);
		}
	}
}
			\end{lstlisting}
		\end{minipage}
		\hfill
		\begin{minipage}{.4\textwidth}
			\begin{lstlisting}[style=condensed]
bool deq(data_t& out) {
	while (true) {
		ptr_t head = Head;
		ptr_t next = head.ptr->next;
		@if (head.age != Head.age) continue;@
		if (head.ptr != Head.ptr) continue;
		if (next.ptr == NULL) {
			return false;
		} else {
			out = next.ptr->data;
			if (CAS(Head, head, next)) {
				ptr_t tail = Tail;
				if (head == tail) {
					CAS(Tail, tail, next);
				}
				@free(head.ptr);@
				return true;
			}
		}
	}
}
			\end{lstlisting}
		\end{minipage}
	\end{minipage}
	\vspace{-4mm}
	\caption{DGLM lock-free queue \cite{DBLP:conf/forte/DohertyGLM04}.}
	\label{appendix:examples:dglm:code}
	\vspace{-45mm}
\end{figure}

\begin{figure}
	\begin{minipage}{.92\textwidth}
		\center
		\begin{minipage}[t]{.52\textwidth}
			\begin{lstlisting}[style=condensed]
S1: atomic {
	/* enq (add new node) */
	assume(Tail.ptr->next.ptr == NULL);
	ptr_t node = create(*);
	Tail.ptr->next = node;
}

S2: atomic {
	/* enq (swing Tail) */
	/* deq (swing Tail) */
	assume(Tail.ptr->next.ptr != NULL);
	Tail.ptr = Tail.ptr->next.ptr;
}
				\end{lstlisting}
			\end{minipage}
			\hfill
			\begin{minipage}[t]{.4\textwidth}
				\begin{lstlisting}[style=condensed]
S3: atomic {
	/* deq */
	assume(Head.ptr->next.ptr != NULL);
	Head.ptr = Head.ptr->next.ptr;
}

S4: atomic { /* skip */ }
			\end{lstlisting}
		\end{minipage}
	\end{minipage}
	\vspace{-4mm}
	\caption{Summary for the DGLM lock-free queue under garbage collection. The algorithm has no summary under explicit memory management.}
	\label{appendix:examples:dglm:summary}
\end{figure}

%!TEX root = ../../main.tex

\section{Semantics}

We provide the detailed semantics for the core language from \cref{sec:programming-model}, as well as semantics for garbage collection and explicit memory management satisfying the requirements discussed in \cref{sec:programming-model,sec:generalization-to-explicit-memory-management}.

% We give detailed semantics for our core language from \cref{sec:programming-model}.
% Moreover, we provide examples of semantics for garbage collection and explicit memory management satisfying the separation (\cref{assumption:sequential-separation}).

%%%%%%%%%%%%%%%%%%%%%%%%%%%%%%%%%%%%%%%%%%%%%%%%%%%%%%%%%%%%%%%%%%%%%%%%%%%%%%
%%%%%%%%%%%%%%%%%%%%%%%%%%%%%%%%%%%% CORE %%%%%%%%%%%%%%%%%%%%%%%%%%%%%%%%%%%%
%%%%%%%%%%%%%%%%%%%%%%%%%%%%%%%%%%%%%%%%%%%%%%%%%%%%%%%%%%%%%%%%%%%%%%%%%%%%%%

\subsection{Core Sequential Semantics}

\Cref{sec:programming-model} introduced the core command language used throughout the paper.
Its semantics is made precise in \cref{appendix:fig:sequential-semantics}.

\begin{figure}
	\newcommand{\mkvspace}{\vspace{4mm}}

	\begin{minipage}{.48\textwidth}
		\begin{prooftree}
			% \rname{seq1}
			\raxiom{
				\state{s}{\conf{\EBNFthread_1}{o}}
				\stepprog
				\state{s'}{\conf{\EBNFthread_1'}{o'}}
			}
			\rconclusion{1}{
				\state{s}{\conf{\EBNFthread_1;\EBNFthread_2}{o}}
				\stepprog
				\state{s'}{\conf{\EBNFthread_1';\EBNFthread_2}{o'}}
			}
		\end{prooftree}
	\end{minipage}
	\hfill
	\begin{minipage}{.48\textwidth}
		\begin{prooftree}
			% \rname{seq2}
			\raxiom{\phantom{()}}
			\rconclusion{1}{
				\state{s}{\conf{\cmdskip;\EBNFthread}{o}}
				\stepprog
				\state{s}{\conf{\EBNFthread}{o}}
			}
		\end{prooftree}
	\end{minipage}
	\mkvspace

	\begin{minipage}{.48\textwidth}
		\begin{prooftree}
			% \rname{choice}
			\raxiom{
				\state{s}{\conf{\EBNFthread_i}{o}}
				\stepprog
				\state{s'}{\conf{\EBNFthread_i'}{o'}}
			}
			\rconclusion{1}{
				\state{s}{\conf{\EBNFthread_1\oplus\EBNFthread_2}{o}}
				\stepprog
				\state{s'}{\conf{\EBNFthread_i'}{o'}}
			}
		\end{prooftree}
	\end{minipage}
	\hfill
	\begin{minipage}{.48\textwidth}
		\begin{prooftree}
			% \rname{atomic1}
			\raxiom{
				\state{s}{\conf{\EBNFthread}{o}}
				\stepprogany
				\state{s'}{\conf{\cmdskip}{o'}}
			}
			\rconclusion{1}{
				\state{s}{\conf{\atomici{\EBNFthread}}{o}}
				\stepprog
				\state{s'}{\conf{\cmdskip}{o'}}
			}
		\end{prooftree}
	\end{minipage}
	\mkvspace

	\begin{minipage}{.48\textwidth}
		\begin{prooftree}
			% \rname{loop1}
			\rconclusion{0}{
				\state{s}{\conf{\EBNFthread^*\!}{o}}
				\stepprog
				\state{s}{\conf{\EBNFthread;\EBNFthread^*}{o}]}
			}
		\end{prooftree}
	\end{minipage}
	\hfill
	\begin{minipage}{.48\textwidth}
	\begin{prooftree}
		% \rname{loop2}
		\rconclusion{0}{
			\state{s}{\conf{\EBNFthread^*\!}{o}}
			\stepprog
			\state{s}{\conf{\cmdskip}{o}]}
		}
	\end{prooftree}
	\end{minipage}
	\mkvspace
	\caption{%
		Semantics of the command language core from \cref{sec:programming-model}.
	}
	\label{appendix:fig:sequential-semantics}
\end{figure}

%%%%%%%%%%%%%%%%%%%%%%%%%%%%%%%%%%%%%%%%%%%%%%%%%%%%%%%%%%%%%%%%%%%%%%%%%%%%%%
%%%%%%%%%%%%%%%%%%%%%%%%%%%%%%%%%%%%%%%%%%%%%%%%%%%%%%%%%%%%%%%%%%%%%%%%%%%%%%
%%%%%%%%%%%%%%%%%%%%%%%%%%%%%%%%%%%%%%%%%%%%%%%%%%%%%%%%%%%%%%%%%%%%%%%%%%%%%%
%%%%%%%%%%%%%%%%%%%%%%%%%%%%%%%%%%%%%%%%%%%%%%%%%%%%%%%%%%%%%%%%%%%%%%%%%%%%%%
%%%%%%%%%%%%%%%%%%%%%%%%%%%%%%%%%%%%%%%%%%%%%%%%%%%%%%%%%%%%%%%%%%%%%%%%%%%%%%
%%%%%%%%%%%%%%%%%%%%%%%%%%%%% GARBAGE COLLECTION %%%%%%%%%%%%%%%%%%%%%%%%%%%%%
%%%%%%%%%%%%%%%%%%%%%%%%%%%%%%%%%%%%%%%%%%%%%%%%%%%%%%%%%%%%%%%%%%%%%%%%%%%%%%
%%%%%%%%%%%%%%%%%%%%%%%%%%%%%%%%%%%%%%%%%%%%%%%%%%%%%%%%%%%%%%%%%%%%%%%%%%%%%%
%%%%%%%%%%%%%%%%%%%%%%%%%%%%%%%%%%%%%%%%%%%%%%%%%%%%%%%%%%%%%%%%%%%%%%%%%%%%%%
%%%%%%%%%%%%%%%%%%%%%%%%%%%%%%%%%%%%%%%%%%%%%%%%%%%%%%%%%%%%%%%%%%%%%%%%%%%%%%
%%%%%%%%%%%%%%%%%%%%%%%%%%%%%%%%%%%%%%%%%%%%%%%%%%%%%%%%%%%%%%%%%%%%%%%%%%%%%%

\subsection{Garbage Collection Semantics}

A possible semantics for garbage collection follows.

\paragraph{Commands.}
In order to present how pointer primitives behave under garbage collection, we consider the following instantiation of $\EBNFprimitive$:
\begin{align*}
	\EBNFprimitive
		~::=&~~
		x = \cmdmalloc \EBNFalt
		x = y  \EBNFalt
		x = y\selsel_i \EBNFalt
		x\selsel_i = y
		\ .
\end{align*}
We assume a typed language where non-pointer values cannot be used or stored as pointers.
Therefore, let $\PSel\subseteq\set{\ptrsel_0,\dots,\ptrsel_n}$ be the selectors of pointer type, and $\DSel=\set{\ptrsel_0,\dots,\ptrsel_n}\setminus\PSel$ the ones of non-pointer type.

\subsubsection{Separation.}
The heap is separated into three partitions: free, owned and shared memory.
A discussion of the individual parts is in order.

\paragraph{Free cells, free records.}
We identify free addresses by $\bot$.
That is, a cell located at address $a$ is free in a heap $h$ if $h(a)=\bot$.
Recall that we may equivalently write $a\notin\dom{h}$.
A record is considered free if so are all its cells.
Formally, we use the predicate $\isfree[h]{a}$ to denote that at address $a$ in heap $h$ sufficiently many cells are available for an object:
\begin{align*}
	\isfree[h]{a}
	:\iff
	\set{a,\dots,a+n}\cap\dom{h}=\varnothing
	\ .
\end{align*}
Note that, due to the semantics of Rule~(\textsc{par}) and the usage of $\bot$ here, freed memory and memory owned by other threads is indistinguishable.

\paragraph{Owned cells.}
Ownership is granted upon allocation.
After a record is allocated the allocating thread owns all the record's cells.
Ownership is lost when the record is made reachable through the shared variables.
Once lost, ownership cannot be reclaimed.

\paragraph{Shared cells.}
Shared are all cells that are neither owned nor free.
A cell becomes shared if it is \emph{published}, \ie, if is made reachable from the shared variables.
If a cell is made unreachable from the shared variables it remains shared nevertheless.

Let $\Ptrshared$ be the shared pointer variables.
Then, the records reachable from the shared (pointer) variables in a heap $h$, denoted by $\Recshared{h}$, are those records collected by the following fixed point:
\begin{align*}
	\mathcal{S}_0 &= \setcond{a}{\exists x\in\Ptrshared.\, h(x)=a}
	\\
	\mathcal{S}_{i+1} &= \setcond{b}{\exists a\in \mathcal{S}_i\,\exists\ptrsel\in\PSel.\; h(a\selsel)=b}
	\ .
\end{align*}
Note that we layout the semantics in such a way that $\bot$ is never reachable from the shared variables.
To do so, it suffices to initialize pointer selectors upon allocation to \code{null}.
This is the case because once a thread acquires a pointer to a cell, that cell remains accessible: if the cell is shared it will remain shared, if the cell is owned then it is owned by the thread holding the reference and it may either remain owned or become and remain shared.
(Hence, ownership violations as discussed for explicit memory management in \cref{sec:generalization-to-explicit-memory-management} cannot occur.)

% Also note that the separation integrated into the concurrent program semantics from \cref{sec:programming-model}, \ie, Rule~(\textsc{par}), makes free memory and memory owned by other threads indistinguishable.

\paragraph{Separating heaps.}
During execution, the boundaries between the shared and owned heap may be adjusted due to, \eg, assignments.
% In order to refer to the consequences of such updates, we introduce a function $\readjust{h,h'}$ which, for two heap $h,h'$, returns a pair of heaps $\pair{s}{o}$ with that reflects the adjustment.
In order to refer to the consequences of such updates, we introduce a function $\readjust{s,o,h'}$ which, for a heap $\soh$, partitioned into a shared part $s$ and an owned part $o$, and an entire updated heap $h'$, provides a partitioning $h'=s'\uplus o'$ into a shared part $s'$ and an owned part $o'$ according to the above intuition of how the shared and owned parts behave.
Formally, $\readjust{s,o,h'}$ sets the domains of $s'$ and $o'$ to
\begin{gather*}
	\dom{s'}=\dom{s}\cup \mathit{Pub}
	\quad\text{and}\quad
	\dom{o'}=\dom{o}\setminus \mathit{Pub}
	\\
	\text{with}\quad
 	\mathit{Pub}:=\setcond{a+0,\dots,a+n}{a\in\Recshared{h'}}\ .
\end{gather*}
This has the effect of moving the published addresses, $\mathit{Pub}$, from the owned heap to the shared heap.
Note that this definition requires $\dom{\soh}=\dom{h'}$.
As we will see later, the only exception to this usage is the transition rule for \code{malloc}.
We will handle this corner case in the transition rule in favor of a more lightweight $\readjust{\cdot}$.

\subsubsection{Semantics.}
Next up are the transition rules for the individual statements of the commands from $\EBNFprimitive$ as instantiated above.

\paragraph{Errors.}
For simplicity we assume that there is a dedicated shared variable $\mathit{fail}$ to indicate execution failures.
We assume that execution is aborted as soon this variable is set.
\begin{prooftree}
	\rname{err}
	\rconclusion{0}{
		\state{s}{\conf{\cmderr}{o}}
		\stepprog
		\state{s[\mathit{fail}\to 1]}{\conf{\cmdskip}{o}}
	}
\end{prooftree}

\paragraph{Allocation.}
A $\cmdmalloc$ allocates new memory for an object.
The memory has to be free prior to the allocation.
The selectors are default-initialized to \code{null}, \ie, $0\in\Nat$.
\begin{prooftree}
	\rname{mallocShared}
	\raxiom{
		\begin{gathered}
			a\in\Nat
			\qquad
			\isfree{a}
			\qquad
			x\in\Ptrshared
			\\
			s'=s[x\to a,a+0\to 0,\dots,a+n\to 0]
		\end{gathered}
	}
	\rconclusion{1}{
		\state{s}{\conf{x=\cmdmalloc}{o}}
		\stepprog
		\state{s'}{\conf{\cmdskip}{o}}
	}
\end{prooftree}
\begin{prooftree}
	\rname{mallocOwned}
	\raxiom{
		\begin{gathered}
			a\in\Nat
			\qquad
			\isfree{a}
			\qquad
			x\notin\Ptrshared
			\\
			o'=s[x\to a,a+0\to 0,\dots,a+n\to 0]
		\end{gathered}
	}
	\rconclusion{1}{
		\state{s}{\conf{x=\cmdmalloc}{o}}
		\stepprog
		\state{s}{\conf{\cmdskip}{o'}}
	}
\end{prooftree}
Note that the allocated memory could be in use by another thread.
Such allocations are prevented in concurrent executions by the semantics of Rule~\textsc{(par)} from \cref{sec:programming-model} as the resulting concurrent state would not be separated.
That is, we do not need care for such \emph{spurious} allocations here in the sequential semantics.

\paragraph{Assignment 1.}
Assignments of the form $x=y$ simply copy the valuation of $y$ to $x$.
The boundary between the shared and owned heap may be readjusted.
\begin{prooftree}
	\rname{assign1}
	\raxiom{\pair{s'}{o'}=\readjust{s,o,\sohup{x\to\sohof{y}}}}
	\rconclusion{1}{
		\state{s}{\conf{x=y}{o}}
		\stepprog
		\state{s'}{\conf{\cmdskip}{o'}}
	}
\end{prooftree}

\paragraph{Assignment 2}
Assignments of the form $x=y\selsel_i$ dereference pointer $y$.
For the assignment to succeed $y$ must be non-null.
As before the shared heap may need readjustment.
\begin{prooftree}
	\rname{assign2-null}
	\raxiom{\sohof{y}=0}
	\rconclusion{1}{
		\state{s}{\conf{x=y\selsel_i}{o}}
		\stepprog
		\state{s}{\conf{\cmderr}{o}}
	}
\end{prooftree}
\begin{prooftree}
	\rname{assign2}
	\raxiom{
		\begin{gathered}
			\sohof{y}=a\neq 0
			\\
			\pair{s'}{o'}=\readjust{s, o, \sohup{x\to\sohof{a\selsel_i}}}
		\end{gathered}
	}
	\rconclusion{1}{
		\state{s}{\conf{x=y\selsel_i}{o}}
		\stepprog
		\state{s'}{\conf{\cmdskip}{o'}}
	}
\end{prooftree}

\paragraph{Assignment 3}
Assignments of the form $x\selsel_i=y$ follow the same rules as the previous assignments.
\begin{prooftree}
	\rname{assign3-null}
	\raxiom{\sohof{x}=0}
	\rconclusion{1}{
		\state{s}{\conf{x\selsel_i=y}{o}}
		\stepprog
		\state{s}{\conf{\cmderr}{o}}
	}
\end{prooftree}
\begin{prooftree}
	\rname{assign3}
	\raxiom{
		\begin{gathered}
			\sohof{x}=a\neq
			\\
			\pair{s'}{o'}=\readjust{s, o, \sohup{a\selsel_i\to\sohof{y}}}
		\end{gathered}
	}
	\rconclusion{1}{
		\state{s}{\conf{x\selsel_i=y}{o}}
		\stepprog
		\state{s'}{\conf{\cmdskip}{o'}}
	}
\end{prooftree}

%%%%%%%%%%%%%%%%%%%%%%%%%%%%%%%%%%%%%%%%%%%%%%%%%%%%%%%%%%%%%%%%%%%%%%%%%%%%%%
%%%%%%%%%%%%%%%%%%%%%%%%%%%%%%%%%%%%%%%%%%%%%%%%%%%%%%%%%%%%%%%%%%%%%%%%%%%%%%
%%%%%%%%%%%%%%%%%%%%%%%%%%%%%%%%%%%%%%%%%%%%%%%%%%%%%%%%%%%%%%%%%%%%%%%%%%%%%%
%%%%%%%%%%%%%%%%%%%%%%%%%%%%%%%%%%%%%%%%%%%%%%%%%%%%%%%%%%%%%%%%%%%%%%%%%%%%%%
%%%%%%%%%%%%%%%%%%%%%%%%%%%%%%%%%%%%%%%%%%%%%%%%%%%%%%%%%%%%%%%%%%%%%%%%%%%%%%
%%%%%%%%%%%%%%%%%%%%%%%%%%%%%%% GENERALIZATION %%%%%%%%%%%%%%%%%%%%%%%%%%%%%%%
%%%%%%%%%%%%%%%%%%%%%%%%%%%%%%%%%%%%%%%%%%%%%%%%%%%%%%%%%%%%%%%%%%%%%%%%%%%%%%
%%%%%%%%%%%%%%%%%%%%%%%%%%%%%%%%%%%%%%%%%%%%%%%%%%%%%%%%%%%%%%%%%%%%%%%%%%%%%%
%%%%%%%%%%%%%%%%%%%%%%%%%%%%%%%%%%%%%%%%%%%%%%%%%%%%%%%%%%%%%%%%%%%%%%%%%%%%%%
%%%%%%%%%%%%%%%%%%%%%%%%%%%%%%%%%%%%%%%%%%%%%%%%%%%%%%%%%%%%%%%%%%%%%%%%%%%%%%
%%%%%%%%%%%%%%%%%%%%%%%%%%%%%%%%%%%%%%%%%%%%%%%%%%%%%%%%%%%%%%%%%%%%%%%%%%%%%%

\subsection{Explicit Memory Management Semantics}

A possible semantics for explicit memory management follows which detects ownership violations as discussed in \cref{sec:generalization-to-explicit-memory-management}.

\paragraph{Commands.}
We consider the following pointer manipulating primitives:
\begin{align*}
	\EBNFprimitive
		~::=&~~
		\cmderr \EBNFalt
		x = \cmdmalloc \EBNFalt
		x = \cmdfree \EBNFalt
		% \mathtt{assume}(x==y) \EBNFalt
		% \mathtt{assume}(x\:!\!=y)
		% \\\EBNFalt&~~
		x = y  \EBNFalt
		x = y\selsel_i \EBNFalt
		x\selsel_i = y
		\ .
\end{align*}
As for the garbage collection case, we use pointer selectors, $\PSel$, non-pointer selectors, $\DSel$, and assume that non-pointer types are not cast to pointer types.

\subsubsection{Separation.}
As for garbage collection, the heap consists of free, owned and shared memory.
The separation is detailed in \cref{sec:generalization-to-explicit-memory-management} and thus not repeated here.

\paragraph{Separating heaps.}
To separate a heap into a shared and an owned part we have to compute the records reachable from the shared variables, $\Ptrshared$, following defined pointers.
We use the symbolic value $\udef$ to denote undefined values of any type.
Then, the addresses of the shared records in a heap $h$ are given by the least fixed point to the following equation:
\begin{align*}
	\mathcal{S}_0 &= \setcond{
		a
	}{
		\exists\: x\in\Ptrshared\!.~
		h(x)=a\neq\udef
	}
	\\
	\mathcal{S}_{i+1} &= \mathcal{S}_{i} \cup \setcond{
		b
	}{
		\exists\: a\in\mathcal{S}_{i}~
		\exists \ptrsel\in\PSel.~
		a\neq\bot \,\wedge\,
		h(a\selsel)=b\neq\udef
	}.
\end{align*}
Let $\mathcal{S}_\mathit{lfp}$ be the least fixed point.
Then, the function $\readjust{h}$ returns a pair $\pair{s'}{o'}$ which resembles $h$ split into a shared heap $s'$ and an owned heap $o'$ according to the computed shared records.
Formally, we have $h=s'\uplus o'$ and 
\begin{align*}
	\dom{s'} =~& \dom{h} ~\cap \bigcup_{a\in\mathcal{S}_\mathit{lfp}} \set{a,\dots,a+n}
	\quad\text{and}\\%\quad
	\dom{o'} =~& \dom{h}\setminus\dom{s'}.
\end{align*}

\paragraph{Ownership violations.}
Recall from \cref{sec:generalization-to-explicit-memory-management} that ownership violations occur by modifying, freeing, or publishing memory owned by other threads.
We handle the former two violations when dealing with assignments and freeing, respectively.
Publishing, however, can happen whenever the boundary between the shared and owned heaps are adjusted.
In such a case a cell owned by another thread is made reachable from the shared variables.
Since those cells appear freed as discussed above, we can recognize an ownership violation due to publishing by checking whether $\bot$ is contained in $\mathcal{S}_\mathit{lfp}$.
For a concise presentation, we let $\readjust{h}$ return $\mathit{err}$ in such occasions.

To make all possible occurrences  of ownership violations visible, we include their detection in the semantics (unlike presented in \cref{sec:generalization-to-explicit-memory-management} where ownership violations due to publishing where detected on top of program steps).

% We introduce the following shortcut:
% \begin{align*}
% 	\isviolation
% 	:\iff&
% 	\exists\: a\in\mathcal{S}_\mathit{lfp}~
% 	\exists \ptrsel\in\mathit{PSel}.~
% 	a\selsel\notin\dom{s}
% 	\ .
% \end{align*}

\subsubsection{Semantics}
The detailed semantics are in order.

\paragraph{Errors.}
As for garbage collection, we assume a dedicated shared variable $\mathit{fail}$ to indicate execution failures.
\begin{prooftree}
	\rname{err}
	\rconclusion{0}{
		\state{s}{\conf{\cmderr}{o}}
		\stepprog
		\state{s[\mathit{fail}\to 1]}{\conf{\cmdskip}{o}}
	}
\end{prooftree}

\paragraph{Allocation.}
A $\cmdmalloc$ allocates new memory for an object.
The memory has to be free prior to the allocation.
The newly allocated memory contains uncontrolled values, \ie, $\udef$.
% This way we allow to draw a sharp boundary between shared and owned memory.
We have:
\begin{prooftree}
	\rname{malloc}
	\raxiom{
		\begin{gathered}
			a\in\Nat
			\qquad
			\isfree{a}
			\qquad
			\pair{s'}{o'}=\readjust{h'}
			\\
			h'=\sohup{x\to a,a+0\to\udef,\dots,a+n\to\udef}
		\end{gathered}
	}
	\rconclusion{1}{
		\state{s}{\conf{x=\cmdmalloc}{o}}
		\stepprog
		\state{s'}{\conf{\cmdskip}{o'}}
	}
\end{prooftree}
Here, $\readjust{\cdot}$ cannot fail unless an ownership violation occurred in the preceding execution.
Hence, we do not include a rule for such cases as we assume that we would have aborted the execution already.

\paragraph{Deallocation.}
A $\cmdfree$ marks a cell available for reallocation.
As discussed for $\isfree[]{\cdot}$ above we do so by setting the memory content to $\bot$.
A $\cmdfree$ removes an entire object from the heap.
The object must not already be free, \ie, we explicitly forbid double free (in particular because it could lead to an ownership violation).
Also, no shared recored is allowed to be freed as for otherwise an allocation is not guaranteed to provide ownership (the allocated cell could still be referenced by a shared object).
\begin{prooftree}
	\rname{free-fail}
	\raxiom{x\in\set{\bot,\udef}}
	\rconclusion{1}{
		\state{s}{\conf{x=\cmdfree}{o}}
		\stepprog
		\state{s}{\conf{\cmderr}{o}}
	}
\end{prooftree}
\begin{prooftree}
	\rname{free-double}
	\raxiom{x\in\dom{\soh}}
	% \raxiom{\isfree{\sohof{x}}}
	\raxiom{\exists\ptrsel\in\PSel.\,\sohof{x\selsel}\notin\dom{\soh}}
	\rconclusion{2}{
		\state{s}{\conf{x=\cmdfree}{o}}
		\stepprog
		\state{s}{\conf{\cmderr}{o}}
	}
\end{prooftree}
\begin{prooftree}
	\rname{free-shared}
	\raxiom{
		\begin{gathered}
			x\in\dom{\soh}
			\qquad
			\neg\isfree{\sohof{x}}
			\\
			\mathit{err}=\readjust{\sohup{a\to\bot,\dots,a+n\to\bot}}
		\end{gathered}
	}
	\rconclusion{1}{
		\state{s}{\conf{x=\cmdfree}{o}}
		\stepprog
		\state{s}{\conf{\cmderr}{o}}
	}
\end{prooftree}
\begin{prooftree}
	\rname{free}
	\raxiom{
		\begin{gathered}
			x\in\dom{\soh}
			\qquad
			\neg\isfree{\sohof{x}}
			\\
			\pair{s'}{o'}=\readjust{\sohup{a\to\bot,\dots,a+n\to\bot}}
		\end{gathered}
	}
	\rconclusion{1}{
		\state{s}{\conf{x=\cmdfree}{o}}
		\stepprog
		\state{s'}{\conf{\cmdskip}{o'}}
	}
\end{prooftree}
% Note that for \textsc{free-shared} one could instead of invoking $\readjust{\cdot}$ alternatively check for $\set{a+0,\dots,a+n}\cap\dom{s}\neq\varnothing$.

% \paragraph{Assume.}
% An assumption continues only if its condition is met.
% Otherwise, execution blocks.
% We forbid comparing variables the valuation of which is $\bot$.
% A value of $\bot$ indicates that a potentially unsafe access was done; one cannot rely on the result.
% \begin{prooftree}
% 	\rname{assume1-fail}
% 	\raxiom{x\notin\dom{\soh} ~\vee~ y\notin\dom{\soh}}
% 	\rconclusion{1}{
% 		\state{s}{\conf{\mathtt{assume}(x==y)}{o}}
% 		\stepprog
% 		\state{s}{\conf{\cmderr}{o}}
% 	}
% \end{prooftree}
% \begin{prooftree}
% 	\rname{assume2-fail}
% 	\raxiom{x\notin\dom{\soh} ~\vee~ y\notin\dom{\soh}}
% 	\rconclusion{1}{
% 		\state{s}{\conf{\mathtt{assume}(x\:!\!=y)}{o}}
% 		\stepprog
% 		\state{s}{\conf{\cmderr}{o}}
% 	}
% \end{prooftree}
% \begin{prooftree}
% 	\rname{assume1}
% 	\raxiom{x\in\dom{\soh}}
% 	\raxiom{y\notin\dom{\soh}}
% 	\rconclusion{2}{
% 		\state{s}{\conf{\mathtt{assume}(x==y)}{o}}
% 		\stepprog
% 		\state{s}{\conf{\cmdskip}{o}}
% 	}
% \end{prooftree}
% \begin{prooftree}
% 	\rname{assume2}
% 	\raxiom{x\in\dom{\soh}}
% 	\raxiom{y\notin\dom{\soh}}
% 	\rconclusion{2}{
% 		\state{s}{\conf{\mathtt{assume}(x\:!\!=y)}{o}}
% 		\stepprog
% 		\state{s}{\conf{\cmdskip}{o}}
% 	}
% \end{prooftree}

\paragraph{Assignment 1.}
Assignments of the form $x=y$ simply copy the valuation of $y$ to $x$.
The boundary between the shared and owned heap may be readjusted.
Consequently, they can fail if an ownership violation due to publishing occurs.
\begin{prooftree}
	\rname{assign1-ov}
	\raxiom{\mathit{err}=\readjust{\sohup{x\to\sohof{y}}}}
	\rconclusion{1}{
		\state{s}{\conf{x=y}{o}}
		\stepprog
		\state{s}{\conf{\cmderr}{o}}
	}
\end{prooftree}
\begin{prooftree}
	\rname{assign1}
	\raxiom{\pair{s'}{o'}=\readjust{\sohup{x\to\sohof{y}}}}
	\rconclusion{1}{
		\state{s}{\conf{x=y}{o}}
		\stepprog
		\state{s'}{\conf{\cmdskip}{o'}}
	}
\end{prooftree}

\paragraph{Assignment 2}
Assignments of the form $x=y\selsel_i$ dereference pointer $y$.
For the assignment to succeed $y$ must be a defined pointer, that is, have a valuation other than $\udef$ and $\bot$.
Otherwise, an unsafe operation is performed which may lead to a segmentation fault.
As before the shared heap may need readjustment.
\begin{prooftree}
	\rname{assign2-seg}
	% \raxiom{y\notin\dom{\soh}}
	\raxiom{y\in\set{\bot,\udef}}
	\rconclusion{1}{
		\state{s}{\conf{x=y\selsel_i}{o}}
		\stepprog
		\state{s}{\conf{\cmderr}{o}}
	}
\end{prooftree}
\begin{prooftree}
	\rname{assign2-ov}
	\raxiom{
		\begin{gathered}
			y\in\dom{\soh}
			\qquad
			\sohof{y}=a
			\\
			\mathit{err}=\readjust{\sohup{x\to\sohof{a\selsel_i}}}
		\end{gathered}
	}
	\rconclusion{1}{
		\state{s}{\conf{x=y\selsel_i}{o}}
		\stepprog
		\state{s}{\conf{\cmderr}{o}}
	}
\end{prooftree}
\begin{prooftree}
	\rname{assign2}
	\raxiom{
		\begin{gathered}
			y\in\dom{\soh}
			\qquad
			\sohof{y}=a
			\\
			\pair{s'}{o'}=\readjust{\sohup{x\to\sohof{a\selsel_i}}}
		\end{gathered}
	}
	\rconclusion{1}{
		\state{s}{\conf{x=y\selsel_i}{o}}
		\stepprog
		\state{s'}{\conf{\cmdskip}{o'}}
	}
\end{prooftree}

\paragraph{Assignment 3}
Assignments of the form $x\selsel_i=y$ dereference pointer $x$.
As above, we require this access to be safe.
More importantly, however, we have to prevent ownership violations.
The previous assignments updated variables.
As such, they could only modify the shared or the owned heap.
Now, we are updating a memory cell.
We have to prevent updates of cells that belong to others.
Since such cells appear freed, we must prevent updates to freed cells.
\begin{prooftree}
	\rname{assign3-seg}
	\raxiom{x\in\set{\bot,\udef}}
	\rconclusion{1}{
		\state{s}{\conf{x\selsel_i=y}{o}}
		\stepprog
		\state{s}{\conf{\cmderr}{o}}
	}
\end{prooftree}
\begin{prooftree}
	\rname{assign3-ov-mod}
	\raxiom{x\in\dom{\soh}}
	% \raxiom{\isfree{\sohof{x}}}
	\raxiom{\sohof{x\selsel_i}\notin\dom{\soh}}
	\rconclusion{2}{
		\state{s}{\conf{x\selsel_i=y}{o}}
		\stepprog
		\state{s}{\conf{\cmderr}{o}}
	}
\end{prooftree}
\begin{prooftree}
	\rname{assign3-ov-pub}
	\raxiom{
		\begin{gathered}
			x\in\dom{\soh}
			\quad~
			\sohof{x}=a
			\quad~
			% \neg\isfree{a}
			a\selsel_i\in\dom{\soh}
			\\
			\mathit{err}=\readjust{\sohup{a\selsel_i\to\sohof{y}}}
		\end{gathered}
	}
	\rconclusion{1}{
		\state{s}{\conf{x\selsel_i=y}{o}}
		\stepprog
		\state{s}{\conf{\cmderr}{o}}
	}
\end{prooftree}
\begin{prooftree}
	\rname{assign3}
	\raxiom{
		\begin{gathered}
			x\in\dom{\soh}
			\qquad
			\sohof{x}=a
			\qquad
			\neg\isfree{a}
			\\
			\pair{s'}{o'}=\readjust{\sohup{a\selsel_i\to\sohof{y}}}
		\end{gathered}
	}
	\rconclusion{1}{
		\state{s}{\conf{x\selsel_i=y}{o}}
		\stepprog
		\state{s'}{\conf{\cmdskip}{o'}}
	}
\end{prooftree}

% \todo[inline]{%
% 	Theorem: ownership violation implies semantics goes to err
% }

\subsubsection{Version counters.}
\label{app:version_counters}

To make the above semantics practical for the usage of version counters we have ensure that threads can accesses these counters at any time.
In particular, it must be possible to access version counters of records owned by other threads.
Intuitively, we make version counters shared no matter if the record they are contained in is shared.

Hereafter, let $\VSel\subseteq\DSel$ be the of selectors which carry version counters and $\NVSel=(\PSel\cup\DSel)\setminus\VSel$ the selectors that do not.
Moreover, let $\Adr\subseteq\Nat$ be the addresses that conform to the alignment restrictions required by version counters\footnote{For version counters to work one has to prevent arbitrary reuse of cells. In particular, certain cells are designated version counters and may not be used for other pointer or data payload. This requires allocations to be aligned, \eg, by the size of records.}.
The following adaptions to the above semantics are required.

\paragraph{Free cells.}
When freeing, the version counters remain valid.
We have to adapt $\isfree[h]{a}$ to account for this.
\begin{align*}
	\isfree[h]{a}
	:\iff
	\forall\ptrsel\in\NVSel.~~
	a\selsel\notin\dom{h}
	\ .
\end{align*}

\paragraph{Separating heaps.}
Version counters remain always shared.
The set of all possible cells carrying version counters is: $\mathcal{M}=\setcond{a\selsel}{a\in\Adr \:\wedge\: \ptrsel\in\VSel}$.
Then, we modify the definition of $\readjust{h}=\pair{s'}{o'}$ to retain version counters in the shared heap:
\begin{align*}
	\dom{s'} =~& \dom{h} ~\cap \left( \mathcal{M} \cup \bigcup_{a\in\mathcal{S}_\mathit{lfp}} \set{a,\dots,a+n} \right)
	\quad\text{and}\\%\quad
	\dom{o'} =~& \dom{h}\setminus\dom{s'}.
\end{align*}

\paragraph{Allocation.}
The values of version counters persist allocations.
If a fresh cell is allocated, then they are initialized to a non-deterministic value.
The following replaces \textsc{(malloc)}:
\begin{prooftree}
	% \rname{tag-malloc}
	\raxiom{
		\begin{gathered}
			a\in\Adr
			\qquad
			\isfree{a}
			\qquad
			\pair{s'}{o'}=\readjust{h'}
			\\
			d_0,\dots,d_n\in\Nat
			\qquad
			\forall\ptrsel_i\in\VSel.~ a\selsel_i\in\dom{s}\implies d_i=s(a\selsel_i)
			\\
			h'=\sohup{x\to a,\forall\ptrsel_i\in\NVSel.~a\selsel_i\to\udef,\forall\ptrsel_i\in\DSel.~a\selsel_i\to d_i}
		\end{gathered}
	}
	\rconclusion{1}{
		\state{s}{\conf{x=\cmdmalloc}{o}}
		\stepprog
		\state{s''}{\conf{\cmdskip}{o'}}
	}
\end{prooftree}

\paragraph{Deallocation.}
Tags are not deleted upon deallocation.
The following replaces \textsc{(free)}:
\begin{prooftree}
	% \rname{tag-free}
	\raxiom{
		\begin{gathered}
			x\in\dom{\soh}
			\qquad
			\neg\isfree{\sohof{x}}
			\\
			\pair{s'}{o'}=\readjust{\sohup{\forall\ptrsel\in\NVSel.~a\selsel\to\bot}}
		\end{gathered}
	}
	\rconclusion{1}{
		\state{s}{\conf{x=\cmdfree}{o}}
		\stepprog
		\state{s'}{\conf{\cmdskip}{o'}}
	}
\end{prooftree}
% Note that it is not necessary to override \textsc{(free-shared)} as deleting version counters for the sake of an ownership violation check does not change the result.

% \include{content/appendix/proofs}
%!TEX root = ../../main.tex

\section{Missing Proofs}

This section gives detailed proofs for the presented theory.
For simplicity we avoid remapping thread identifiers among executions of different programs.
Therefore, without loss of generality, we assume during the proofs throughout this section that the following holds:
\begin{align*}
	\cfinit{\Program\inpar\SummaryStar}&=\cfinit{\Program}\uplus\cfinit{\SummaryStar}
	\\\text{and}\quad
	\cfinit{\Thread\inpar\SummaryStar}&=\set{0\mapsto\conf{\Thread}{\emp}}\uplus\cfinit{\SummaryStar}
	\ .
\end{align*}

%%%%%%%%%%%%%%%%%%%%%%%%%%%%%%%%%%%%%%%%%%%%%%%%%%%%%%%%%%%%%%%%%%%%%%%%%%%%%%
%%%%%%%%%%%%%%%%%%%%%%%%%%%%%%%%% AUXILIARIES %%%%%%%%%%%%%%%%%%%%%%%%%%%%%%%%
%%%%%%%%%%%%%%%%%%%%%%%%%%%%%%%%%%%%%%%%%%%%%%%%%%%%%%%%%%%%%%%%%%%%%%%%%%%%%%

\subsection{Auxiliaries}

\begin{fact}
	\label{new:thm:aux:stateless-steps-are-iterable}
	Let $\Summary$ be a program and $\Thread\in\Summary$ on of its threads.
	Then, for every step $\state{s}{\conf{\Thread}{\emp}}\stepprog\state{s'}{\conf{\cmdskip}{\emp}}$ of $\Thread$,
	we have $\state{s}{\cfinit{\SummaryStar}}\stepprogany\state{s'}{\cfinit{\SummaryStar}}$.
\end{fact}

\begin{lemma}
	\label{new:thm:aux:main-aux}
	% Let $\Summary$ be a stateless program and $\Thread\in\Summary$ one of its threads.
	% Then, for every effect $\pair{s}{s'}\in\effects{\Thread\inpar\SummaryStar}$ with $s\neq s'$, there is some thread $\Thread'\in\Summary$ with $\state{s}{\conf{\Thread'}{\emp}}\stepprog\state{s'}{\conf{\cmdskip}{\emp}}$.
	Let $\Summary$ be a stateless program and $\Thread\in\Summary$ one of its threads.
	Then, for every $\pair{s}{s'}\in\effects{\Thread\inpar\SummaryStar}$, there is some $\Thread'\in\Summary$ with $\state{s}{\conf{\Thread'}{\emp}}\stepprog\state{s'}{\conf{\cmdskip}{\emp}}$.
\end{lemma}

\begin{corollary}
	\label{new:thm:aux:stateless-effects-give-repeated-executions}
	Let $\Summary$ be a stateless program.
	Then, for every $\pair{s}{s'}\in\effects{\SummaryStar}$, we have $\state{s}{\cfinit{\SummaryStar}}\stepprogany\state{s'}{\cfinit{\SummaryStar}}$.
\end{corollary}

\begin{corollary}
	\label{new:thm:aux:stateless-does-not-care-for-duplicated-threads}
	Let $\Summary$ be a stateless program and let $\Thread\in\Summary$ be one of its threads.
	Then, $\effects{\Thread\inpar\SummaryStar}\subseteq\effects{\SummaryStar}$.
\end{corollary}

\begin{corollary}
	\label{new:thm:aux:stateless-effects-are-stateless-steps}
	Let $\Summary$ be a stateless program.
	Then, for every $\pair{s}{s'}\in\effects{\SummaryStar}$, there is some $\Thread\in\Summary$ with $\state{s}{\conf{\Thread}{\emp}}\stepprog\state{s'}{\conf{\cmdskip}{\emp}}$.
\end{corollary}

%
% Proof fact
%
\begin{proof}[\Cref{new:thm:aux:stateless-steps-are-iterable}]
	Let $\Summary$ be a program and $\Thread\in\Summary$ one of its threads.
	Furthermore, let $\state{s}{\conf{\Thread}{\emp}}\stepprog\state{s'}{\conf{\cmdskip}{\emp}}$ be a valid $\Thread$-step.
	Then, the semantics give:
	\begin{align*}
		\state{s}{\conf{\Thread^*}{\emp}}
		\stepprog
		&\state{s}{\conf{\Thread;\Thread^*}{\emp}}
		\stepprog
		\state{s'}{\conf{\cmdskip;\Thread^*}{\emp}}\\
		\stepprog
		&\state{s'}{\conf{\Thread^*}{\emp}}
		\ .
	\end{align*}
	Now, let $i\in\Tid$ be such that $\cfinit{\SummaryStar}(i)=\conf{\Thread^*}{\emp}$.
	Such $i$ must exist because we have $\Thread\in\Summary$.
	Then, using the above sequence of $\Thread^*$ steps gives a valid sequence:
	\begin{align*}
		\state{s}{\cfinit{\SummaryStar}}
		=~
		&\state{s}{\cfinit{\SummaryStar}[i\to\conf{\Thread^*}{\emp}]} \\
		\stepprogany
		&\state{s'}{\cfinit{\SummaryStar}[i\to\conf{\Thread^*}{\emp}]}\\
		=~
		&\state{s'}{\cfinit{\SummaryStar}}
		\ .
		\tag*{\qed}
	\end{align*}
	Note that this sequence is valid by the semantics, because at any time the local heaps of all threads are $\emp$.
	Hence, concurrent steps are possible.
	This concludes the claim.
\end{proof}

%
% Proof Lemma
%
\begin{proof}[\Cref{new:thm:aux:main-aux}]
	Let $\Summary$ be a stateless program and $\Thread\in\Summary$ one of its threads.
	We now show that the executions of $\Thread\inpar\SummaryStar$ are of of a particular form.
	
	\begin{description}[labelwidth=6mm,leftmargin=8mm,itemindent=0mm]
		\item[IB:]
			The empty execution of $\Thread\inpar\SummaryStar$ reaches only its initial state, $\state{\sinit}{\cfinit{\Thread\inpar\SummaryStar}}$.
			The initial state of $\SummaryStar$ is $\state{\sinit}{\cfinit{\SummaryStar}}$.
			Hence, $\SummaryStar$ reaches $\sinit$ in zero steps: $\state{\sinit}{\cfinit{\SummaryStar}}\stepprogsome{0}\state{\sinit}{\cfinit{\SummaryStar}}$.
			Moreover, the initial configurations are of the form $\cfinit{\Thread\inpar\SummaryStar}(0)=\conf{\Thread}{\emp}$ and $\cfinit{\Thread\inpar\SummaryStar}(i)=\conf{\Thread_i}{\emp}$ for every $i\in\dom{\cfinit{\SummaryStar}}$ and some thread $\Thread_i\in\Thread$.
			\\[-2mm]

		\item[IH:]
			For every $\init{\Thread\inpar\SummaryStar}\stepprogany\state{s}{\cf}$ we have
			\begin{compactenum}[(a)]
			 	\item $\init{\SummaryStar}\stepprogany\state{s}{\cfinit{\SummaryStar}}$
			 	\item $\cf(0)\in\set{\conf{\Thread}{\emp},\conf{\cmdskip}{\emp}}$
			 	\item $\cf(i)\in\set{\conf{\cmdskip}{\emp},\conf{\cmdskip;\Thread_i^*}{\emp},\conf{\Thread_i^*}{\emp},\conf{\Thread_i;\Thread_i^*}{\emp}}$ for every $i,\Thread_i$ with $\cfinit{\SummaryStar}(i)=\conf{\Thread_i^*}{\emp}$
			 	\\[-2mm]
			\end{compactenum}

		\item[IS:]
			Consider now $\init{\Thread\inpar\SummaryStar}\stepprogany\state{s}{\cf}\stepprog[i]\state{s'}{\cf'}$ for some $i$.
			The semantics give $\cf(j)=\cf'(j)$ for all $j\neq i$.
			Hence, for proof obligations (b) and (c) it suffices to show that $\cf'(i)$ has the desired form.
			We do a case distinction on $i$.
			\\[-2mm]

			\textit{Case $i = 0$.}
			By induction (b) we must have $\cf(0)\!=\!\conf{\Thread}{\emp}$ with $\Thread\not\equiv\cmdskip$ as for otherwise the last step would not be possible.
			The induction hypothesis~(a), furthermore, provides $s\in\Reach{\SummaryStar}$.
			Then, we have $\cf'(0)=\conf{\cmdskip}{\emp}$ because $\Thread$ is stateless as $\Summary$ is.
			Hence, $\cf'$ has the appropriate form and the step $\state{s}{\conf{\Thread}{\emp}}\stepprog\state{s'}{\conf{\cmdskip}{\emp}}$ is valid.
			Now, \cref{new:thm:aux:stateless-steps-are-iterable} together with induction hypothesis (a) yields:
			$\init{\SummaryStar}\stepprogany\state{s}{\cfinit{\SummaryStar}}\stepprogany\state{s'}{\cfinit{\SummaryStar}}$.
			This concludes the case.
			\\[-2mm]

			\textit{Case $i \neq 0$.}
			We cannot have $\cf(i)=\conf{\cmdskip}{\emp}$ as the last step would not be possible otherwise.
			For $\cf(i)\in\set{\conf{\cmdskip;\Thread_i^*}{\emp},\conf{\Thread_i^*}{\emp}}$ we immediately get $s=s'$ and thus $\init{\SummaryStar}\stepprogany\state{s'}{\cfinit{\SummaryStar}}$ by induction.
			Moreover, the semantics immediately give the desired form of $\cf$.
			So it remains to consider $\cf(i)=\conf{\Thread_i;\Thread_i^*}{\emp}$ with $\Thread_i\not\equiv\cmdskip$.
			Due to $s\in\Reach{\SummaryStar}$ by induction (a) and statelessness of $\Thread_i\in\Summary$ we have $\state{s}{\conf{\Thread_i}{\emp}}\stepprog\state{s'}{\conf{\cmdskip}{\emp}}$.
			Hence, we get the desired $\cf'(i)=\conf{\cmdskip;\Thread_i^*}{\emp}$.
			Now, \cref{new:thm:aux:stateless-steps-are-iterable} together with induction hypothesis (a) yields:
			$\init{\SummaryStar}\stepprogany\state{s}{\cfinit{\SummaryStar}}\stepprogany\state{s'}{\cfinit{\SummaryStar}}$.
			This concludes the case.
	\end{description}

	\noindent
	Let $\pair{s}{s'}\in\effects{\Thread\inpar\SummaryStar}$ be some effect.
	It stems from an execution of the form $\init{\Thread\inpar\SummaryStar}\stepprogany\state{s}{\cf}\stepprog[i]\state{s'}{\cf'}$ for some $i$.
	Analogous to the induction step from above, we conclude that there is some thread $\Thread'$, namely the one with identifier $i$, that can perform the valid step $\state{s}{\conf{\Thread'}{\emp}}\stepprog\state{s}{\conf{\cmdskip}{\emp}}$.
	\qed
\end{proof}

%
% Proof Corollary
%
\begin{proof}[\Cref{new:thm:aux:stateless-effects-give-repeated-executions}]
	Let $\Summary$ be a stateless program and $\pair{s}{s'}\in\effects{\SummaryStar}$ one of its effects.
	% If $s=s'$ nothing needs to be shown.
	% Otherwise, we invoke \cref{new:thm:aux:main-aux} and get some $\Thread\in\Summary$ with $\state{s}{\conf{\Thread}{\emp}}\stepprog\state{s'}{\conf{\cmdskip}{\emp}}$.
	We invoke \cref{new:thm:aux:main-aux} to get some $\Thread\in\Summary$ with $\state{s}{\conf{\Thread}{\emp}}\stepprog\state{s'}{\conf{\cmdskip}{\emp}}$.
	Then, we conclude applying \cref{new:thm:aux:stateless-steps-are-iterable} onto this step.
	\qed
\end{proof}

%
% Proof Corollary
%
\begin{proof}[\Cref{new:thm:aux:stateless-does-not-care-for-duplicated-threads}]
	Let $\Summary$ be a stateless program and $\Thread\in\Summary$ one of its threads.
	Let $\pair{s}{s'}\in\effects{\Thread\inpar\SummaryStar}$ be some effect.
	% Since $\pair{s}{s}\in\effects{\SummaryStar}$ is trivially true due to the iteration of $\Summary$, consider $s\neq s'$.
	The effect stems from an execution of the form $\init{\Thread\inpar\SummaryStar}\stepprogany\state{s}{\cf}\stepprog\state{s'}{\cf'}$.
	According to the induction from the Proof of \cref{new:thm:aux:main-aux} we have $\init{\SummaryStar}\stepprogany\state{s}{\cfinit{\SummaryStar}}$.
	Invoking \cref{new:thm:aux:main-aux} on the effect $\pair{s}{s'}$ yields some $\Thread'\in\Summary$ with $\state{s}{\conf{\Thread'}{\emp}}\stepprog\state{s'}{\conf{\cmdskip}{\emp}}$.
	Applying \cref{new:thm:aux:stateless-steps-are-iterable} to this step then gives $\state{s}{\cfinit{\SummaryStar}}\stepprogany\state{s'}{\cfinit{\SummaryStar}}$.
	Altogether, we now have the desired $\init{\SummaryStar}\stepprogany\state{s}{\cfinit{\SummaryStar}}\stepprogany\state{s'}{\cfinit{\SummaryStar}}$.
	\qed
\end{proof}

%
% Proof Corollary
%
\begin{proof}[\Cref{new:thm:aux:stateless-effects-are-stateless-steps}]
	Let $\Summary$ be a stateless program and let $\pair{s}{s'}\in\effects{\SummaryStar}$ be one its effects.
	Now, pick some thread $\Thread'\in\Summary$.
	Then, \cref{new:thm:aux:main-aux} applied to $\Summary$, $\Thread'$ and $\pair{s}{s'}$ gives some $\Thread$ with $\state{s}{\conf{\Thread}{\emp}}\stepprog\state{s'}{\conf{\cmdskip}{\emp}}$ as desired.
	\qed
\end{proof}

%%%%%%%%%%%%%%%%%%%%%%%%%%%%%%%%%%%%%%%%%%%%%%%%%%%%%%%%%%%%%%%%%%%%%%%%%%%%%%
%%%%%%%%%%%%%%%%%%%%%%%%%%%%%%%%% MODULARITY %%%%%%%%%%%%%%%%%%%%%%%%%%%%%%%%%
%%%%%%%%%%%%%%%%%%%%%%%%%%%%%%%%%%%%%%%%%%%%%%%%%%%%%%%%%%%%%%%%%%%%%%%%%%%%%%

\subsection{Proof of \Cref{new:thm:effect-inclusion-lifting}}

Let $\Program,\Summary$ be two programs such that $\Summary$ is stateless and $\effects{\Thread\inpar\SummaryStar}\subseteq\effects{\SummaryStar}$ holds for all threads $\Thread\in\Program$.
We show: $\effects{\Program}\subseteq\effects{\Program\inpar\SummaryStar}\subseteq\effects{\SummaryStar}$.

For the first inclusion, note that the effects of $\Program$ are trivially included in the effects of $\Program\inpar\SummaryStar$ by the definition of effects; letting $\SummaryStar$ in $\Program\inpar\SummaryStar$ stutter resembles exactly $\Program$.
Hence, we focus on the second inclusion hereafter.

Towards the desired result, we show that the threads from $\Program$ cannot distinguish whether they run in parallel with threads from $\Program$ or with $\SummaryStar$.
% In other words, on could say that we replace steps from $\Program$ with steps from $\SummaryStar$ until a single thread is left to run in parallel with $\SummaryStar$.
% In this process, statelessness of $\Summary$ is crucial: it allows to repeatedly use $\SummaryStar$ without its iterations influencing each other.
% Moreover, statelessness allows $\Summary$ to be applicable repeatedly.
Formally, we show that, for every execution $\init{\Program\inpar\SummaryStar}\stepprogany\state{s}{\cf}$ and for every $k\in\dom{\cf},~ \Thread\in\Program\inpar\SummaryStar$ with $\cfinit{\Program\inpar\SummaryStar}(k)=\conf{\Thread}{\emp}$, there is another execution with the same $\Thread$-configuration but the remaining threads simulated by $\SummaryStar$: $\init{\Thread\inpar\SummaryStar}\stepprogany\state{s}{\cfinit{\Thread\inpar\SummaryStar}[0\to\cf(k)]}$.
Note that in the new execution we always let $\SummaryStar$ come back to its initial configuration to allow further iterations.
We proceed by induction over the structure of $\Program\inpar\SummaryStar$ execution.

\begin{description}[labelwidth=6mm,leftmargin=8mm,itemindent=0mm]
	\item[IB:]
		The empty execution reaches only the initial program state.
		For $\Program\inpar\SummaryStar$ this is $\state{\sinit}{\cfinit{\Program\inpar\SummaryStar}}$.
		By definition, we have $\cfinit{\Program\inpar\SummaryStar}(k)=\conf{\Thread}{\emp}$ for every $k,\Thread$.
		Similarly, we have $\cfinit{\Thread\inpar\SummaryStar}(k)=\conf{\Thread}{\emp}$ by definition.
		Hence, the empty execution of $\Thread\inpar\SummaryStar$ reaches a state of the desired form.
		\\[-2mm]

	\item[IH:]
		For every $\init{\Program\inpar\SummaryStar}\stepprogany\state{s}{\cf}$ and every $k\in\dom{\cf},~\Thread\in\Program\inpar\SummaryStar$ with $\cfinit{\Program\inpar\SummaryStar}(k)=\conf{\Thread}{\emp}$ there is $\init{\Thread\inpar\SummaryStar}\stepprogany\state{s}{\cfinit{\Thread\inpar\SummaryStar}[0\to\cf(k)]}$.
		\\[-2mm]

	\item[IS:]
		Consider now $\init{\Program\inpar\SummaryStar}\stepprogany\state{s}{\cf}\stepprog[i]\state{s'}{\cf'}$ for some $i$.
		% Let $k\in\dom{\cf}$ and let $\Thread_k\in\Program\inpar\SummaryStar$ be arbitrary with $\cfinit{\Program\inpar\SummaryStar}(k)=\conf{\Thread_k}{\emp}$.
		The semantics give $\state{s}{\cf(i)}\stepprog\state{s'}{\cf'(i)}$.
		Let $\Thread_i\in\Program\inpar\SummaryStar$ with $\cfinit{\Program\inpar\SummaryStar}(i)=\conf{\Thread_i}{\emp}$.
		Then, the induction hypothesis for $i$ and $\Thread_i$ gives:
		\begin{align*}
			\init{\Thread_i\inpar\SummaryStar}
			\stepprogany
			&\state{s}{\cfinit{\Thread_i\inpar\SummaryStar}[0\to\cf(i)]}
			\\
			\stepprog[0]
			&\state{s'}{\cfinit{\Thread_i\inpar\SummaryStar}[0\to\cf'(i)]}.
		\end{align*}
		The last step of thread $0$ is valid here for the following reason.
		The state $\state{s}{\cf(i)}$ is separated as for otherwise its step to $\state{s'}{\cf'(i)}$ would not be possible.
		Hence, \cref{assumption:sequential-separation} states that also \state{s}{\cf(i)} is separated.
		Since the owned heap of all threads $j$, $j\neq 0$, are empty, the concurrent result state of the execution is separated.
		% Note that the last step is possible since the owned heaps of all threads except for thread $0$ are guaranteed to be $\emp$.
		% Hence, the last step of thread $0$ is valid by Rule~(\textsc{par}) because the destination state is separated by \cref{assumption:sequential-separation}.
		\\[-2mm]

		Now, consider some arbitrary $k\in\dom{\cf}$ and some arbitrary $\Thread_k\in\Program\inpar\SummaryStar$ with $\cfinit{\Program\inpar\SummaryStar}(k)=\conf{\Thread_k}{\emp}$.
		In order to conclude the induction, we show that an execution of the desired form exists for $k,\Thread_k$.
		For the case $k=i$ we are done because of the above execution.
		So consider $k\neq i$.
		\\[-2mm]

		By definition the above $\Thread_i\inpar\SummaryStar$ execution gives $\pair{s}{s'}\in\effects{\Thread_i\inpar\SummaryStar}$.
		This yields $\pair{s}{s'}\in\effects{\SummaryStar}$ due to the premise (in case of $\Thread_i\in\Program$) and \cref{new:thm:aux:stateless-does-not-care-for-duplicated-threads} (in case of $\Thread_i\in\SummaryStar$).
		Hence, $\state{s}{\cfinit{\SummaryStar}}\stepprogany\state{s'}{\cfinit{\SummaryStar}}$ follows from \cref{new:thm:aux:stateless-effects-give-repeated-executions}.
		Together with the induction hypothesis for $k,\Thread_k$ we get
		\begin{align*}
			\init{\Thread_k\inpar\SummaryStar}
			\stepprogany
			&\state{s}{\cfinit{\Thread_k\inpar\SummaryStar}[0\to\cf(k)]}
			\\
			\stepprogany
			&\state{s'}{\cfinit{\Thread_k\inpar\SummaryStar}[0\to\cf(k)]}
		\end{align*}
		which is a valid execution because the owned heaps of the intermediate $\SummaryStar$ states are always $\emp$ as shown by \cref{new:thm:aux:main-aux} (we skip a formal discussion as it would simply repeat the arguments of that lemma).
		This concludes the induction because we have $\cf(k)=\cf'(k)$ due to $k\neq i$ by the semantics.
\end{description}

\noindent
Now, consider an effect $\pair{s}{s'}\in\effects{\Program\inpar\SummaryStar}$.
The effect stems from some execution of the form $\init{\Program\inpar\SummaryStar}\stepprogany\state{s_1}{\cf_1}\stepprog\state{s_1'}{\cf_1'}$.
By the above induction we know that there is another execution of the form $\init{\Thread\inpar\SummaryStar}\stepprogany\state{s_2}{\cf_2}\stepprog\state{s_2'}{\cf_2'}$ with $\Thread\in\Program\inpar\SummaryStar$.
This gives $\pair{s}{s'}\in\effects{\Thread\inpar\SummaryStar}$, as discussed in the induction step,  (due to the premise and \cref{new:thm:aux:stateless-does-not-care-for-duplicated-threads}), and concludes the claim.
\qed

%%%%%%%%%%%%%%%%%%%%%%%%%%%%%%%%%%%%%%%%%%%%%%%%%%%%%%%%%%%%%%%%%%%%%%%%%%%%%%
%%%%%%%%%%%%%%%%%%%%%%%%%%%%%%%%%% REACHSET %%%%%%%%%%%%%%%%%%%%%%%%%%%%%%%%%%
%%%%%%%%%%%%%%%%%%%%%%%%%%%%%%%%%%%%%%%%%%%%%%%%%%%%%%%%%%%%%%%%%%%%%%%%%%%%%%

\subsection{Proof of \Cref{new:thm:fixed-point-reachset}}

Let $\Summary$ be a summary of $\Program$.
We show the equality in two parts:
\begin{compactenum}[(i)]
	\item[(i)]
		for every $s\in\ReachFP$ there is some $\Thread\in\Program$ such that $s\in\Reach{\Thread\inpar\Summary}$, and
	\item[(ii)]
		for every $\Thread\in\Program$ and every $s\in\Reach{\Thread\inpar\Program}$ we have $s\in\ReachFP$.
\end{compactenum}

\subsubsection{Part~(i).}

We show that every shared heap explored by the fixed point there is an execution that reaches this shared heap, too.
We proceed by induction over the iterations of the fixed point.

\begin{description}[labelwidth=6mm,leftmargin=8mm,itemindent=0mm]
	\item[IB:]
		Let $\state{s}{\cf}\in X_0$.
		By definition, $s=\sinit$ and $\cf=\conf{\Thread}{\emp}$ for some $\Thread\in\Program$.
		The initial configuration of $\Thread$ in $\Thread\inpar\SummaryStar$ is $\cfinit{\Thread\inpar\SummaryStar}(0)=\conf{\Thread}{\emp}$.
		Hence, we have reachability in zero steps: $\init{\Thread\inpar\SummaryStar}\stepprogsome{0}\state{s}{\cfinit{\Thread\inpar\SummaryStar}[0\to\cf]}$.
		Similarly, we get $\init{\SummaryStar}\stepprogsome{0}\conf{s}{\cfinit{\SummaryStar}}$.
		\\[-2mm]

	\item[IH:]
		For every $\state{s}{\cf}\in X_i$ we have $\init{\SummaryStar}\stepprogany\conf{s}{\cfinit{\SummaryStar}}$ and there is some $\Thread\in\Program$ with $\init{\Thread\inpar\SummaryStar}\stepprogany\state{s}{\cfinit{\Thread\inpar\SummaryStar}[0\to\cf]}$.
		\\[-2mm]

	\item[IS:]
		Consider now some $\state{s'}{\cf'}\in X_{i+1}$.
		By definition we have $\state{s'}{\cf'}\in X_i$ or $\state{s'}{\cf'}\in\Post{X_i}$ or $\state{s'}{\cf'}\in\Env{X_i}$.
		In the first case, we conclude by induction.
		In the second case, $\state{s'}{\cf'}\in\Post{X_i}$, there is some $\state{s}{\cf}$ with $\state{s}{\cf}\stepprog\state{s'}{\cf'}$.
		By \cref{assumption:sequential-separation} we know that $\state{s'}{\cf'}$ is separated.
		So by induction we get some $\Thread\in\Program$ with:
		\[
			\init{\Thread\inpar\SummaryStar}
			\stepprogany
			\state{s}{\cfinit{\Thread\inpar\SummaryStar}[0\to\cf]}
			\stepprog[0]
			\state{s'}{\cfinit{\Thread\inpar\SummaryStar}[0\to\cf']}
			% \ .
		\]
		where the last step is valid because separation holds as only the owned heap of $\cf'$ may be non-empty.
		This execution yields then the effect $\pair{s}{s'}\in\effects{\Thread\inpar\SummaryStar}$.
		Since $\Summary$ is a summary of $\Program$ we get $\pair{s}{s'}\in\effects{\SummaryStar}$.
		Thus, \cref{new:thm:aux:stateless-effects-give-repeated-executions} together with the induction hypothesis gives $\init{\SummaryStar}\stepprogany\conf{s}{\cfinit{\SummaryStar}}\stepprogany\conf{s'}{\cfinit{\SummaryStar}}$.
		This concludes this case.
		\\[-2mm]

		Last, consider $\state{s'}{\cf'}\in\Env{X_i}$.
		By definition there is some $s,\cf''$ such that $\state{s}{\cf'}\in X_i$ and $\state{s}{\cfinit{\Summary}}\stepprog_k\state{s'}{\cf''}$.
		Let $\Thread_k$ be some thread such that $\cfinit{\Summary}=\conf{\Thread_k}{\emp}$.
		By induction, $s\in\Reach{\SummaryStar}$, thus $\cf''(k)=\conf{\cmdskip}{\emp}$ because $\Summary$ is stateless.
		Then, together with the induction, we get:
		\begin{align*}
			\init{\SummaryStar}
			\stepprogany
			&\conf{s}{\cfinit{\SummaryStar}}
			=
			\conf{s}{\cfinit{\SummaryStar}[k\to\conf{\Thread_k^*}{\emp}]}\\
			\stepprog\:\,
			&\conf{s}{\cfinit{\SummaryStar}[k\to\conf{\Thread_k;\Thread_k^*}{\emp}]}\\
			\stepprog\:\,
			&\conf{s'}{\cfinit{\SummaryStar}[k\to\conf{\cmdskip;\Thread_k^*}{\emp}]}
		\end{align*}
		 % $\init{\SummaryStar}\stepprogany\conf{s}{\cfinit{\SummaryStar}}\stepprog\state{s'}{\cf'''}$ for some $\cf'''$.
		Hence, $\pair{s}{s'}\in\effects{\SummaryStar}$.
		Now, \cref{new:thm:aux:stateless-effects-give-repeated-executions} yields ${\cfinit{\SummaryStar}}\stepprogany\conf{s'}{\cfinit{\SummaryStar}}$.
		So, ${\init{\SummaryStar}}\stepprogany\conf{s'}{\cfinit{\SummaryStar}}$ holds by induction.
		Similar to the above, we invoke the induction again and get some $\Thread\in\Program$ and construct the following execution:
		\begin{align*}
			\init{\Thread\inpar\SummaryStar}
			\stepprogany
			&\conf{s}{\cfinit{\Thread\inpar\SummaryStar}}
			=
			\conf{s}{\cfinit{\Thread\inpar\SummaryStar}[k\to\conf{\Thread_k^*}{\emp}]}\\
			\stepprog\:\,
			&\conf{s}{\cfinit{\Thread\inpar\SummaryStar}[k\to\conf{\Thread_k;\Thread_k^*}{\emp}]}\\
			\stepprog\:\,
			&\conf{s'}{\cfinit{\Thread\inpar\SummaryStar}[k\to\conf{\cmdskip;\Thread_k^*}{\emp}]}\\
			\stepprog\:\,
			&\conf{s'}{\cfinit{\Thread\inpar\SummaryStar}[k\to\conf{\Thread_k^*}{\emp}]}\\
			=~~~
			&\conf{s'}{\cfinit{\Thread\inpar\SummaryStar}}
		\end{align*}
		Note that the above execution is valid because $\state{s'}{\cf'}$ is separated by definition of $\Env{\cdot}$ together with $\state{s'}{\cf'}\in\Env{X_i}$, and by the fact that only the owned heap of $\cf'$ may be non-empty.
		This concludes the induction.
		% Appending this sequence of steps to the execution provided by the induction hypothesis concludes the claim. % (the thread configuration of $\Thread$ is not changed).
\end{description}

\noindent
Now, consider some $s\in\ReachFP$.
By definition there is some $\cf$ such that $\state{s}{\cf}\in X_k$.
Then, the above induction provides $s\in\Reach{\Thread\inpar\SummaryStar}$ for some $\Thread\in\Program$.
\qed

\subsubsection{Part~(ii).}

Let $\Thread\in\Program$ be some thread.
We show that every reachable shared heap of $\Thread\inpar\SummaryStar$ is explored by the fixed point.
We proceed by induction over the structure of executions.

\begin{description}[labelwidth=6mm,leftmargin=8mm,itemindent=0mm]
	\item[IB:]
		The empty execution reaches only $\sinit$, and by definition $\state{\sinit}{\cfinit{\Thread}}\in X_k$.
		\\[-2mm]

	\item[IH:]
		For every execution $\init{\Thread\inpar\SummaryStar}\stepprogany\state{s}{\cf}$ we have
		$\state{s}{\cfinit{\Thread}}\in X_k$.
		\\[-2mm]

	\item[IS:]
		Consider now $\init{\Thread\inpar\SummaryStar}\stepprogany\state{s}{\cf}\stepprog\state{s'}{\cf'}$.
		This execution yields the effect $\pair{s}{s'}\in\effects{\Thread\inpar\SummaryStar}$.
		Since $\Summary$ is a summary of $\Program$, we can use the effect inclusion to get $\pair{s}{s'}\in\effects{\SummaryStar}$.
		Then, \cref{new:thm:aux:stateless-effects-are-stateless-steps} yields some $\Thread'\in\Summary$ with $\state{s}{\conf{\Thread'}{\emp}}\stepprog\state{s'}{\conf{\cmdskip}{\emp}}$.
		By definition together with the semantics we get: $\state{s}{\cfinit{\Summary}}\stepprog\state{s'}{\cf''}$ for some $\cf''$.
		And by induction we have $\state{s}{\cfinit{\Thread}}\in X_k$.
		Hence, $\state{s}{\cfinit{\Thread}}\in\Env{X_k}\subseteq X_k$ concludes the induction.
\end{description}

\noindent
Consider now some $s\in\Reach{\Thread\inpar\SummaryStar}$.
By definition there is an execution of the form $\init{\Thread\inpar\SummaryStar}\stepprogany\state{s}{\cf}$.
The above induction now provides $\state{s}{\cfinit{\Thread}}\in X_k$.
Hence, the desired $s\in\ReachFP$ holds by definition.
\qed

%%%%%%%%%%%%%%%%%%%%%%%%%%%%%%%%%%%%%%%%%%%%%%%%%%%%%%%%%%%%%%%%%%%%%%%%%%%%%%
%%%%%%%%%%%%%%%%%%%%%%%%%%%%%%%%%%% SOUND1 %%%%%%%%%%%%%%%%%%%%%%%%%%%%%%%%%%%
%%%%%%%%%%%%%%%%%%%%%%%%%%%%%%%%%%%%%%%%%%%%%%%%%%%%%%%%%%%%%%%%%%%%%%%%%%%%%%

\subsection{Proof of \Cref{new:thm:summary-implies-soundness}}

Let $\Program$ be a programs and $\Summary$ its summary.
This provides the effect inclusion: $\effects{\Thread\inpar\SummaryStar}\subseteq\effects{\SummaryStar}$.
Then, by \cref{new:thm:effect-inclusion-lifting} we get $\effects{\Program\inpar\SummaryStar}\subseteq\effects{\SummaryStar}$.
Because $\effects{\Program}\subseteq\effects{\Program\inpar\SummaryStar}$ is always true, we have $\effects{\Program}\subseteq\effects{\SummaryStar}$.
By definition this yields $\Reach{\Program}\subseteq\Reach{\SummaryStar}$ and concludes the first inclusion.
For the remaining equality we proceed as follows.
First, note that $\Reach{\SummaryStar}\subseteq\bigcup_{\Thread\in\Program}\Reach{\Thread\inpar\SummaryStar}$ is always true.
Hence, due to \cref{new:thm:fixed-point-reachset}, we have $\Reach{\SummaryStar}\subseteq\ReachFP$.
Now, again by effect inclusion, we have $\Reach{\Thread\inpar\SummaryStar}\subseteq\Reach{\SummaryStar}$.
Consequently, we have $\bigcup_{\Thread\in\Program}\Reach{\Thread\inpar\SummaryStar}\subseteq\bigcup_{\Thread\in\Program}\Reach{\SummaryStar}\subseteq\Reach{\SummaryStar}$.
Hence, \cref{new:thm:fixed-point-reachset} yields $\ReachFP\subseteq\Reach{\SummaryStar}$.
Altogether, this gives the desired equality: $\Reach{\SummaryStar}=\ReachFP$.
\qed

%%%%%%%%%%%%%%%%%%%%%%%%%%%%%%%%%%%%%%%%%%%%%%%%%%%%%%%%%%%%%%%%%%%%%%%%%%%%%%
%%%%%%%%%%%%%%%%%%%%%%%%%%%%%%%%%%% SOUND2 %%%%%%%%%%%%%%%%%%%%%%%%%%%%%%%%%%%
%%%%%%%%%%%%%%%%%%%%%%%%%%%%%%%%%%%%%%%%%%%%%%%%%%%%%%%%%%%%%%%%%%%%%%%%%%%%%%

\subsection{Proof of \Cref{new:thm:checks-iff-summarization}}

We show the individual implications.
\begin{compactenum}[(i)]
	\item[(i)]
		If $X_k$ satisfies the checks \eqref{new:fp:check:mimic-effects} and \eqref{new:fp:check:stateless}, then $\Summary$ is a summary of $\Program$.
	\item[(ii)]
		If $\Summary$ is a summary of $\Program$, then $X_k$ satisfies \eqref{new:fp:check:mimic-effects} and \eqref{new:fp:check:stateless}.
\end{compactenum}

\subsubsection{Part~(i).}

Let $X_k$ be the fixed point from \cref{sec:effect-summaries} and let it satisfy the checks \eqref{new:fp:check:mimic-effects} and \eqref{new:fp:check:stateless}.
First, we show an auxiliary statement, then, we show statelessness, and, last, we show the effect inclusion.

\paragraph{Auxiliary induction.}
Let $\Thread\in\Program$ be some program.
We show that the reachable states of the program $\Thread\inpar\SummaryStar$ are explored by the fixed point.
We proceed by induction over executions.
Note that we showed a similar property before.
However, we cannot reuse it because it relies on the property we want to show, namely $\Summary$ being a summary of $\Program$.

\begin{description}[labelwidth=6mm,leftmargin=8mm,itemindent=0mm]
	\item[IB:]
		The empty execution reaches only $\sinit$.
		By definition, $\state{\sinit}{\cfinit{\Thread}}\in X_k$.
		\\[-2mm]

	\item[IH:]
		For every execution $\init{\Thread\inpar\SummaryStar}\stepprogany\state{s}{\cf}$ we have
		$\state{s}{\cf(0)}\in X_k$
		and moreover
		$\cf(i)\in\set{\conf{\cmdskip}{\emp},\conf{\cmdskip;\Thread_i^*}{\emp},\conf{\Thread_i^*}{\emp},\conf{\Thread_i;\Thread_i^*}{\emp}}$ for every $i,\Thread_i$ with $\cfinit{\SummaryStar}(i)=\conf{\Thread_i^*}{\emp}$.
		\\[-2mm]

	\item[IS:]
		Consider now $\init{\Thread\inpar\SummaryStar}\stepprogany\state{s}{\cf}\stepprog[i]\state{s'}{\cf'}$ for some $i$.
		Due to the semantics we have $\cf(j)=\cf'(j)$ for all $j\neq i$.
		Hence, the second proof obligation boils down to showing that $\cf'(i)$ has the desired form.
		First, consider $i=0$.
		This immediately satisfies the second proof obligation.
		Due to the semantics we have $\state{s}{\cf(0)}\stepprog\state{s'}{\cf'(0)}$.
		Hence, $\state{s}{\cf(0)}\in\Post{X_k}\subseteq X_k$ holds by induction.
		Consider now $i\neq 0$.
		We do a case distinction on $\cf(i)$ according to the induction hypothesis.
		\\[-2mm]

		\textit{Case $\cf(i)=\conf{\cmdskip}{\emp}$.}
		This case cannot apply as it does not allow for the step to $\state{s'}{\cf'}$.
		\\[-2mm]

		\textit{Case $\cf(i)=\conf{\cmdskip;\Thread^*}{\emp}$.}
		We get $s=s'$ and $\cf'(i)=\conf{\Thread^*}{\emp}$ by the semantics.
		Hence, $\cf(i)$ has the desired form.
		Moreover, $\state{s'}{\cf'(0)}\in X_k$ holds by induction.
		\\[-2mm]

		\textit{Case $\cf(i)=\conf{\Thread^*}{\emp}$.}
		Similarly to the previous case, we immediately conclude because we have $s=s'$ and $\cf'(i)\in\set{\conf{\cmdskip}{\emp},\conf{\Thread;\Thread^*}{\emp}}$.
		\\[-2mm]

		\textit{Case $\cf(i)=\conf{\Thread;\Thread^*}{\emp}$ with $\Thread\neq\cmdskip$.}
		Due to the semantics we get $\state{s}{\conf{\Thread}{\emp}}\stepprog\state{s'}{\cf''}$ for some $\cf''$.
		Hence, $\state{s}{\cfinit{\Summary}(i)}\stepprog\state{s'}{\cf''}$ by definition.
		Moreover, the induction provides $\state{s}{\cf(0)}\in X_k$.
		So we can invoke \eqref{new:fp:check:stateless} and get $\cf''=\conf{\cmdskip}{\emp}$.
		As a consequence, we get $\cf'(i)=\conf{\cmdskip;\Thread^*}{\emp}$ of the desired form.
		Moreover, by definition of the fixed point, we have $\state{s'}{\cf(0)}\in\Env{X_k}$.
		Note here that $\state{s'}{\cf(0)}$ is separated because $\state{s'}{\cf'}$ is separated by the semantics and we have $\cf(0)=\cf'(0)$.
		This concludes the induction because of $\Env{X_k}\subseteq X_k$.
\end{description}

\paragraph{Statelessness.}
Consider some thread $\Thread\in\Summary$, some reachable heap $s\in\Reach{\SummaryStar}$, and some transition $\state{s}{\cfinit{\Thread}}\stepprog\state{s'}{\cf''}$.
By definition we have $s\in\Reach{\Thread'\inpar\SummaryStar}$ for some $\Thread'\in\Program$.
So the auxiliary induction yields some $\state{s}{\cf}\in X_k$.
Moreover, we have $\cfinit{\Summary}(i)=\conf{\Thread}{\emp}$ for some $i$ by definition.
Hence, \eqref{new:fp:check:stateless} yields $\cf''=\conf{\cmdskip}{\emp}$ since $\cfinit{\Thread}=\conf{\Thread}{\emp}$.
That is, $\Summary$ is stateless by definition.

\paragraph{Effect inclusion.}
Let $\Thread\in\Program$ be some thread.
Towards the effect inclusion we show that the reachable shared heaps of $\Thread\inpar\SummaryStar$ are explored by $\SummaryStar$ alone.
We proceed by induction over executions.

\begin{description}[labelwidth=6mm,leftmargin=8mm,itemindent=0mm]
	\item[IB:]
		The empty execution reaches only $\sinit$.
		By definition, $\init{\SummaryStar}=\state{\sinit}{\cfinit{\SummaryStar}}$.
		\\[-2mm]

	\item[IH:]
		For every execution $\init{\Thread\inpar\SummaryStar}\stepprogany\state{s}{\cf}$ we have $\init{\SummaryStar}\stepprogany\state{s}{\cfinit{\SummaryStar}}$.
		\\[-2mm]

	\item[IS:]
		Consider now $\init{\Thread\inpar\SummaryStar}\stepprogany\state{s}{\cf}\stepprog[i]\state{s'}{\cf'}$ for some $i$.
		Consider $s\neq s'$ as we immediately conclude by induction otherwise.
		There are two cases.
		\\[-2mm]

		\textit{Case $i\neq 0$.}
		Due to the auxiliary induction together with the semantics, $\cf(i)$ must be of the form $\conf{\Thread_i;\Thread_i^*}{\emp}$ for some $\Thread_i\in\Summary$.
		So we have a valid step as follows: $\state{s}{\conf{\Thread_i}{\emp}}\stepprog\state{s'}{\cf''}$ for some $\cf''$.
		As shown above, $\Summary$ is stateless; and so is $\Thread_i\in\Summary$.
		Hence, $\cf''=\conf{\cmdskip}{\emp}$ must hold because $s\in\Reach{\SummaryStar}$ by induction.
		Now, \cref{new:thm:aux:stateless-steps-are-iterable} yields $\state{s}{\cfinit{\SummaryStar}}\stepprogany\state{s'}{\cfinit{\SummaryStar}}$.
		This concludes the case by induction.
		\\[-2mm]

		\textit{Case $i= 0$.}
		By the semantics we have $\state{s}{\cf(0)}\stepprog\state{s'}{\cf'(0)}$.
		The auxiliary induction gives $\state{s}{\cf(0)}\in X_k$.
		Hence, \eqref{new:fp:check:mimic-effects} gives the step $\state{s}{\cfinit{\Summary}}\stepprog[k]\state{s'}{\cf''}$ for some $k,\cf''$.
		By statelessness of $\Summary$ and $s\in\Reach{\SummaryStar}$ by induction we must have $\cf''(k)=\conf{\cmdskip}{\emp}$.
		That is, there is some $\Thread_k\in\Summary$ with $\state{s}{\conf{\Thread_k}{\emp}}\stepprog\state{s'}{\conf{\cmdskip}{\emp}}$.
		Now, we invoke \cref{new:thm:aux:stateless-steps-are-iterable} again which yields $\state{s}{\cfinit{\SummaryStar}}\stepprogany\state{s'}{\cfinit{\SummaryStar}}$.
		This concludes the case by induction.
\end{description}

\noindent
Consider now some effect $\pair{s}{s'}\in\effects{\Thread\inpar\SummaryStar}$.
This effect stems from an execution of the form $\init{\Thread\inpar\SummaryStar}\stepprogany\state{s}{\cf}\stepprog\state{s'}{\cf'}$.
By the above induction we have $\init{\SummaryStar}\stepprogany\state{s}{\cfinit{\SummaryStar}}$.
Moreover, with an analogous reasoning as for the induction step we get some thread $\Thread_k\in\Summary$ with $\state{s}{\conf{\Thread'}{\emp}}\stepprog\state{s'}{\conf{\cmdskip}{\emp}}$ and $\cfinit{\SummaryStar}(k)=\conf{\Thread_k^*}{\emp}$.
Hence, the following is a valid execution:
\begin{align*}
	\init{\SummaryStar}
	\stepprogany
	&\state{s}{\cfinit{\SummaryStar}}
	=~
	\state{s}{\cfinit{\SummaryStar}[k\to\conf{\Thread_k*}{\emp}]}\\
	\stepprog~
	&\state{s}{\cfinit{\SummaryStar}[k\to\conf{\Thread_k;\Thread_k*}{\emp}]}\\
	\stepprog~
	&\state{s'}{\cfinit{\SummaryStar}[k\to\conf{\cmdskip;\Thread_k*}{\emp}]}
	\ .
\end{align*}
That is, $\pair{s}{s'}\in\effects{\SummaryStar}$.
This establishes the effect inclusion and together with statelessness from above shows that $\Summary$ is a summary of $\Program$.
\qed

\subsubsection{Part~(ii).}

Let $\Summary$ be a summary of $\Program$.
First we show that \eqref{new:fp:check:stateless} holds, afterwards we tackle \eqref{new:fp:check:mimic-effects}.

\paragraph{\eqref{new:fp:check:stateless}}
Consider some $s\in\ReachFP$ and some transition $\state{s}{\cfinit{\Summary}}\stepprog\state{s'}{\cf'}$.
By definition, we have $\cfinit{\SummaryStar}=\conf{\Thread}{\emp}$ and thus $\state{s}{\conf{\Thread}{\emp}}\stepprog\state{s'}{\cf'}$.
Moreover, by the induction of the Proof of \cref{new:thm:fixed-point-reachset}, Part~(i), we have $s\in\Reach{\SummaryStar}$.
Hence, we immediately conclude $\cf'=\conf{\cmdskip}{\emp}$ because $\Summary$ is stateless.

\paragraph{\eqref{new:fp:check:mimic-effects}}
Consider some $\state{s}{\cf}\in X_k$ and some transition $\state{s}{\cf}\stepprog\state{s'}{\cf'}$.
By the induction of the Proof of \cref{new:thm:fixed-point-reachset}, Part~(i), there is some thread $\Thread\in\Program$ with $\cfinit{\Thread\inpar\SummaryStar}\stepprogany\state{s}{\cfinit{\Thread\inpar\SummaryStar}[0\to\cf]}$.
So by the semantics, we have \[\cfinit{\Thread\inpar\SummaryStar}\stepprogany\state{s}{\cfinit{\Thread\inpar\SummaryStar}[0\to\cf]}\stepprog\state{s'}{\cfinit{\Thread\inpar\SummaryStar}[0\to\cf']}\]
because $\state{s'}{\cf'}$ is separated by \cref{assumption:sequential-separation} and thus also the resulting state of the execution, $\state{s'}{\cfinit{\Thread\inpar\SummaryStar}[0\to\cf']}$, is separated.
That is, $\pair{s}{s'}\in\effects{\Thread\inpar\SummaryStar}$ holds by definition.
Hence, we get $\pair{s}{s'}\in\effects{\SummaryStar}$ because of the effect inclusion (which is provided by the fact that $\Summary$ is a summary of $\Program$).
Now, \cref{new:thm:aux:stateless-effects-are-stateless-steps} yields some $\Thread'\in\Summary$ with $\state{s}{\conf{\Thread'}{\emp}}\stepprog\state{s'}{\conf{\cmdskip}{\emp}}$.
Hence, we have $\state{s}{\cfinit{\Summary}}\stepprog\state{s'}{\cf''}$ for some $\cf$.
This concludes the \eqref{new:fp:check:mimic-effects}.
\qed

\end{document}